\documentclass[11pt]{article}
\usepackage{palatino}
\usepackage{pgfplots} 
\usepackage{tikz, tikzpeople}
\usetikzlibrary{calc,intersections,decorations.pathmorphing,decorations.markings,patterns,arrows.meta,positioning}
\usetikzlibrary{angles,quotes}
\usetikzlibrary{arrows.meta}
\usetikzlibrary{arrows}
\usetikzlibrary{shapes.multipart}
\usetikzlibrary{fit}
\usetikzlibrary{shapes.geometric,positioning}
\usetikzlibrary{calc}
\usepackage{bm}
\usepackage{mathtools}
\usepackage{xfrac}
\usepackage{graphicx}
\usepackage[linesnumbered,ruled,vlined]{algorithm2e}
\usepackage{cite}
\usepackage{amsmath,amssymb,amsfonts, amsthm,  mathrsfs}
\usepackage[linesnumbered,ruled,vlined]{algorithm2e}
\usepackage{algpseudocode}
\usepackage{textcomp}
\usepackage{xcolor}
\usepackage{braket}
\usepackage{enumitem}
\usepackage{cancel}
\usepackage{hyperref}
\usepackage{stfloats}
\usepackage{xparse}

\usetikzlibrary{patterns}
\newtheorem{lemma}{Lemma}
\newtheorem{theorem}{Theorem}

\newtheorem{remark}{Remark}
\newtheorem{corollary}{Corollary}

\newlength{\leftstackrelawd}
\newlength{\leftstackrelbwd}
\def\leftstackrel#1#2{\settowidth{\leftstackrelawd}%
{${{}^{#1}}$}\settowidth{\leftstackrelbwd}{$#2$}%
\addtolength{\leftstackrelawd}{-\leftstackrelbwd}%
\leavevmode\ifthenelse{\lengthtest{\leftstackrelawd>0pt}}%
{\kern-.5\leftstackrelawd}{}\mathrel{\mathop{#2}\limits^{#1}}}

\newcommand{\Sato}{\mathchoice
  {\mathrm{\scriptscriptstyle  Sato}} 
  {\mathrm{\scriptscriptstyle  Sato}} 
  {\mathrm{\scriptscriptstyle  Sato}} 
  {\mathrm{\scriptscriptstyle  Sato}} 
}

\newcommand{\KS}{\mathchoice
  {\mathrm{\scriptscriptstyle  KS}} 
  {\mathrm{\scriptscriptstyle  KS}} 
  {\mathrm{\scriptscriptstyle  KS}} 
  {\mathrm{\scriptscriptstyle  KS}} 
}

\newcommand{\KSp}{\mathchoice
  {\mathrm{\scriptscriptstyle KS+}} 
  {\mathrm{\scriptscriptstyle KS+}} 
  {\mathrm{\scriptscriptstyle KS+}} 
  {\mathrm{\scriptscriptstyle KS+}} 
}

\newcommand{\BLEC}{\mathchoice
  {\mathrm{\scriptscriptstyle BLEC}} 
  {\mathrm{\scriptscriptstyle BLEC}} 
  {\mathrm{\scriptscriptstyle BLEC}} 
  {\mathrm{\scriptscriptstyle BLEC}} 
}

\newcommand{\BSSC}{\mathchoice
  {\mathrm{\scriptscriptstyle BSSC}} 
  {\mathrm{\scriptscriptstyle BSSC}} 
  {\mathrm{\scriptscriptstyle BSSC}} 
  {\mathrm{\scriptscriptstyle BSSC}} 
}

\newcommand{\mN}{\mathsf{N}} 
\newcommand{\mZ}{\mathsf{Z}} 

\NewDocumentCommand{\mP}{e{_}}{%
  \mathsf{P}%
  \IfValueT{#1}{_{\!#1}}%
}

\newcommand{\mS}{\mathsf{S}} 
\newcommand{\mQ}{\mathsf{Q}}

\newcommand{\NS}{\mathchoice
  {\mathrm{\scriptscriptstyle NS}} 
  {\mathrm{\scriptscriptstyle NS}} 
  {\mathrm{\scriptscriptstyle NS}} 
  {\mathrm{\scriptscriptstyle NS}} 
}

\DeclareMathOperator{\conv}{conv}

\allowdisplaybreaks

\title{The Capacity Region of the Broadcast Channel with Non-Signaling Assistance}
\author{Yuhang Yao, Syed A. Jafar\\
{\small Nhu Department of Electrical Engineering and Computer Science}\\
{\small University of California Irvine, Irvine, CA 92697}\\
{\small \it Email: \{yuhangy5, syed\}@uci.edu}
}

\date{}

\begin{document}
\maketitle

\begin{abstract}
The capacity region of the $K$-user discrete memoryless broadcast channel is fully characterized when non-signaling (NS) assistance is available to the transmitter and all $K$ receivers. The NS-assisted capacity region is shown to coincide with Sato's region, i.e., the region defined by sum-rate bounds over all subsets of messages, where each bound corresponds to full cooperation among that subset of receivers under a worst-case joint channel law consistent with the marginal channels.
\end{abstract}

\section{Introduction}

The  ultimate limits of reliable data transmission over a (in general, multiuser) communication channel are of central interest in information theory. Such limits are represented by various notions of channel \emph{capacity}.\footnote{By capacity we refer in this work only to Shannon capacity formulations where  probability of error converges to zero in the limit of large blocklengths.} By relaxing external constraints, e.g., by allowing unlimited pre-shared resources between encoders and decoders, the capacity of a channel captures an intrinsic property of the channel itself. 

The \emph{classical} capacity of a channel allows only classical  resources, i.e., classical correlations to be shared in advance. It is known that stronger correlations, e.g., the non-local correlations enabled by pre-shared quantum entanglement between the encoders and decoders, can help achieve rates above the classical capacity in some cases \cite{Quek_Shor, leditzky2020playing, seshadri2023separation, Yao_Jafar_QEACWS}. In fact, recent work has noted that the gains can be quite large (cf. exponential or unbounded multiplicative gains in \cite{Yao_Jafar_Unlock}). The search for the ultimate limits thus requires a broadening of the scope of pre-shared resources. 

Arguably the model with the broadest scope is one that allows unlimited non-signaling (NS) correlations to be shared in advance. A resource is non-signaling if it does not constitute a communication channel by itself, i.e., the resource \emph{by itself} does not allow any communication among the distributed parties. It is known that NS-resources include quantum resources, and the inclusion is strict. The notion of capacity that allows encoders and decoders to share any non-signaling correlations in advance is called the \emph{NS-assisted capacity}.

NS-assisted capacity characterizations are thus far available primarily for single user point-to-point channels. Matthews showed in \cite{matthews2012linear} that NS assistance does not increase the capacity of a point-to-point channel. For a point-to-point channel with state, the capacity with channel state information at the transmitter (CSIT) is characterized in \cite{Yao_Jafar_CSITTP} if the CSIT is non-causal, and in \cite{Yao_Jafar_NSCWS} if the CSIT is causal. In both cases the NS-assisted capacity is the same as the classical capacity of the corresponding channel where the state is also provided to the receiver. Remarkably, NS-assistance can produce an unbounded multiplicative gain over classical capacity in such settings \cite[Example 1 (fading dirt channel)]{Yao_Jafar_CSITTP}.

Exact capacity (region) characterizations with NS-assistance have mostly remained  elusive thus far in multiuser settings. Fawzi and Ferm\'{e}  explored the NS-assisted capacity of the multiple access channel (MAC) in \cite{fawzi2024MAC}, and the broadcast channel (BC) in \cite{fawzi2024broadcast}. Multiplicative gains in capacity from NS-assistance  that have been noted thus far in MAC settings are relatively modest. Somewhat surprisingly, however, \cite{Yao_Jafar_NS_DoF} provides examples of $K$-user broadcast channels (for any number of users $K$) with  multiplicative gain approaching a factor of $K$. The NS-assisted capacity region of the broadcast channel is known in some special cases. For example, it is shown in \cite{fawzi2024broadcast} that NS-assistance does not improve the capacity region of a BC if the NS resource is shared only among the receivers, or if the BC is deterministic. If NS-assistance is only available between the transmitter and one of the users (say User $1$) of a $K=2$ user BC, then the capacity region is shown in \cite{Yao_Jafar_NS_DoF} to be the same as the classical capacity region for the setting where the desired message of User $2$ is made available in
advance as side-information to User $1$. For the $2$-user semi-deterministic BC, the capacity region is found in \cite{Yao_Jafar_NS_DoF}, and turns out to be the same as if NS-assistance was available only between the transmitter and the non-deterministic user. Beyond these special cases, the capacity region of the broadcast channel with NS-assistance is not known in general. This open problem is our focus in this paper.

As our main result, we fully characterize the capacity region of the $K$-user discrete memoryless broadcast channel when non-signaling (NS) assistance is available to the transmitter and all $K$ receivers. Specifically, we show (Theorem \ref{thm:NSBC_Kuser} in Section \ref{sec:NSBC_Kuser}) that the NS-assisted capacity region is the same as Sato's region, i.e., the region defined by sum-rate bounds over all subsets of messages, where each bound corresponds to full cooperation among that subset of receivers under a worst-case joint channel law consistent with the marginal channels. 

The essential idea underlying this result is the `\emph{authentication solution}' from our previous work in \cite{Yao_Jafar_NSCWS} (see also \cite{OTPmodel}) on the NS-assisted capacity of a point-to-point channel with causal CSIT. A NS-assisted coding scheme for a $K$-user BC can be represented as a $(K+1)$-partite box (a NS-box) with an input and output for each of the $K+1$ parties (the transmitter and the $K$ users). Intuitively, an authentication solution can be summarized as follows. The NS-box provides a secret \emph{key} (independent of the messages) to the transmitter. This key is transmitted on the broadcast channel. Each user obtains  its own channel output and provides it as an input to its own side of the NS-box. The NS-box  tests (authenticates) if the input provided by the user is jointly typical with the transmitted key under the distribution induced by the channel between the transmitter and that user. If the authentication is successful then the NS-box outputs the correct message for the authenticated user. Otherwise, with a certain probability it outputs a randomly chosen incorrect message. These probabilities are chosen carefully to \emph{exactly} satisfy the non-signaling constraints. Fortunately, because the channel's action ensures that atypical outputs only occur with vanishing probability, these choices can be tuned without hurting the asymptotic reliability. Finally, the terms corresponding to Sato's bound for each subset of users naturally emerge in such a scheme as the exponents in the upper bound on the probability of \emph{accidental typicality,} i.e., the probability that arbitrary inputs from that subset of users are all authenticated (Lemma \ref{lem:typical_off_K} in Section \ref{sec:preliminaries}).

\section{Preliminaries}
\subsection{Sets}
Let $\mathbb{N}$ denote the set of positive integers. Let $\mathbb{R}_{\geq 0}$ and $\mathbb{R}_{> 0}$ denote the sets of non-negative and positive real numbers, respectively. 
For $n\in \mathbb{N}$, the notation $[n]$ denotes the set of integers $\{1,2,\ldots, n\}$, and $x^n$ denotes $(x_1,x_2,\ldots, x_n)$.
We write $\mathcal{S} \subseteq \mathcal{T}$ if $\mathcal{S}$ is a subset of $\mathcal{T}$, and $\mathcal{S} \subset \mathcal{T}$ if $\mathcal{S}$ is a proper subset of $\mathcal{T}$.
For $\mathcal{S}=\{i_1,i_2,\ldots,i_m\}\subset\mathbb{N}$, the notation $(A_i)_{i\in\mathcal{S}}$ denotes the tuple $(A_{i_1},A_{i_2},\ldots,A_{i_m})$, where $i_1<i_2<\cdots<i_m$. 
We use $\mathbb{I}(\cdot)$ to denote the indicator function, which equals $1$ if the predicate is true and $0$ otherwise.

\subsection{Probability}
A probability mass function (pmf) on a finite-cardinality discrete set $\mathcal{X}$ is a function $\mP_{X}\colon \mathcal{X} \to \mathbb{R}_{\geq 0}$ satisfying $\sum_{x\in \mathcal{X}} \mP_X(x) = 1$. The set of all pmfs on $\mathcal{X}$ is denoted by $\mathcal{P}(\mathcal{X})$. 
The notation $X \sim P$ indicates that the random variable $X$ is distributed according to the pmf $\mP$.
Let $\mathcal{X} \times \mathcal{Y}$ denote the Cartesian product of $\mathcal{X}$ and $\mathcal{Y}$, and let $\mathcal{X}^n$ denote the $n$-fold Cartesian product of $\mathcal{X}$. 
A joint pmf $\mP_{X_1,\ldots,X_n}$ is simply a pmf on the Cartesian product $\mathcal{X}_1\times\cdots\times\mathcal{X}_n$.
A conditional pmf $\mP_{X\mid Y}$ associated with an ordered pair of discrete sets $(\mathcal{X},\mathcal{Y})$ is a function $\mP_{X\mid Y} \colon \mathcal{X}\times \mathcal{Y} \to \mathbb{R}_{\geq 0}$ such that, for every $y\in \mathcal{Y}$, $\sum_{x}\mP_{X\mid Y}(x\mid y) = 1$. The set of all conditional pmfs associated with $(\mathcal{X}, \mathcal{Y})$ is denoted by $\mathcal{P}(\mathcal{X} \mid \mathcal{Y})$. 
We define $\mP_{Y\mid X=x}$ as the pmf in $\mathcal{P}(\mathcal{Y})$ for which $\mP_{Y\mid X=x}(y) = \mP_{Y\mid X}(y\mid x)$.
Given a joint pmf $\mP_{XY}$, the marginal pmfs $\mP_X$ and $\mP_Y$, as well as the conditional pmfs $\mP_{Y\mid X}$ and $\mP_{X\mid Y}$, are defined in the usual way. Whenever $\mP_X(x)=0$, the conditional pmf $\mP_{Y\mid X=x}$ may be chosen arbitrarily from $\mathcal{P}(\mathcal{Y})$.
The notation $\mP_X^{\otimes n}(x^n)$ denotes the product pmf $\mP_X^{\otimes n}(x^n)=\prod_{i=1}^n \mP_{X}(x_i)$, and $\mP_{Y\mid X}^{\otimes n}(y^n \mid x^n) = \prod_{i=1}^n \mP_{Y\mid X}(y_i\mid x_i)$.
We use $\Pr(E)$ to denote the probability of a random event $E$, and $\mathbb{E}[X]$ to denote the expectation of $X$.

\subsection{Information measures}
For $X\sim \mP_X$, let $H_{\mP_{X}}(X) = -\sum_{x} \mP_X(x) \log_2 \mP_X(x) = -\mathbb{E}_{X\sim \mP_X}[\log_2(\mP_X(X))]$ denote the entropy of $X$. For $(X,Y)\sim \mP_{XY}$, let $I_{\mP_{XY}}(X;Y) = \sum_{(x,y)}\mP_{XY}(x,y) \log_2 \frac{\mP_{XY}(x,y)}{\mP_{X}(x)\mP_Y(y)}$ denote the mutual information of $X,Y$. The subscript indicating the underlying distribution may be omitted when it is clear from the context. 
For $\mP_X, \mQ_X \in \mathcal{P}(\mathcal{X})$, let $D(\mP_X\Vert \mQ_X)$ denote the relative entropy (KL divergence), i.e., $D(\mP_X\Vert \mQ_X) = \sum_{x\in \mathcal{X}}\mP_X(x)\log_2\frac{\mP(x)}{\mQ(x)}$. 
For $\mP_{XY}, \mQ_{XY} \in \mathcal{P}(\mathcal{X} \times \mathcal{Y})$, let $D(\mP_{Y|X}\Vert \mQ_{Y|X} ~|~\mP_X) = \sum_{x\in \mathcal{X}}\mP_X(x) D(\mP_{Y\mid X=x} \Vert \mQ_{Y\mid X=x})$ denote the conditional relative entropy.

Let $\kappa \in \mathbb{N}$. For $\kappa$ conditional pmfs $\mP_{Y_1\mid X}, \mP_{Y_2\mid X}, \ldots, \mP_{Y_\kappa\mid X}$, where $\mP_{Y_k\mid X} \in \mathcal{P}(\mathcal{Y}_k\mid \mathcal{X}), \forall k\in [\kappa]$,  define
\begin{equation}
\begin{aligned}
	&\Gamma(\mP_{Y_1\mid X}, \mP_{Y_2\mid X}, \ldots, \mP_{Y_\kappa\mid X}) \triangleq \Big\{ \mQ_{Y_1Y_2\cdots Y_\kappa\mid X} \in \mathcal{P}(\mathcal{Y}_1\times \cdots \times \mathcal{Y}_\kappa \mid \mathcal{X}) : \\
	&\hspace{7cm} \mQ_{Y_k\mid X} = \mP_{Y_k\mid X}, \forall k\in [\kappa] \Big\}
\end{aligned}
\end{equation}
as the subset that consists of the conditional pmfs $\mQ_{Y_1Y_2\cdots Y_K\mid X}$ whose $K$ marginal pmfs for $Y_1, Y_2,\ldots, Y_K$ conditioned on $X$ are equal to $\mP_{Y_1\mid X}, \mP_{Y_2\mid X}, \ldots, \mP_{Y_K\mid X}$, respectively.
Given $\mP_X \in \mathcal{P}(\mathcal{X})$ and $\mP_{Y_k\mid X} \in \mathcal{P}(\mathcal{Y}_k\mid \mathcal{X})$ for all $k\in[\kappa]$, define the function
\begin{align} \label{eq:def_JK}
	J(\mP_X, (\mP_{Y_1\mid X}, \ldots, \mP_{Y_\kappa\mid X})) \triangleq \min_{\mQ_{Y_1\cdots Y_\kappa\mid X}\in \Gamma(\mP_{Y_1\mid X}, \ldots, \mP_{Y_\kappa\mid X})} I_{\mP_X\mQ_{Y_1\cdots Y_\kappa\mid X}}(X;Y_1,\ldots, Y_{\kappa}).
\end{align}
The minimum in \eqref{eq:def_JK} exists because $\Gamma(\mP_{Y_1\mid X}, \mP_{Y_2\mid X},\ldots, \mP_{Y_K\mid X})$ is a nonempty compact polytope ($\mathcal{X},\mathcal{Y}_1,\dots,\mathcal{Y}_\kappa$ are finite-cardinality discrete sets) and mutual information is continuous.
For the special case of $\kappa=1$, the definition reduces to a mutual information, i.e., $J(\mP_X,\mP_{Y_1\mid X}) = I_{\mP_X\mP_{Y_1\mid X}}(X_1;Y)$.

\subsection{Types and strong typicality}
For a sequence $x^n\in \mathcal{X}^n$, its empirical pmf (type) $\widehat{\mP}_{x^n}\in \mathcal{P}(\mathcal{X})$ is defined by
\begin{align} \label{eq:def_type}
	\widehat{\mP}_{x^n}(x) \triangleq \frac{1}{n} \sum_{i=1}^n \mathbb{I}(x_i=x), \qquad \forall x\in \mathcal{X}.
\end{align} 
Given $\mP_X\in \mathcal{P}(\mathcal{X})$, $\epsilon \in (0,1)$ and $n\in \mathbb{N}$, define the (strongly) typical set
\begin{align} \label{eq:def_typical_set}
	\mathcal{T}_{\epsilon}^{(n)}(\mP_X) \triangleq \big\{ x^n \in \mathcal{X}^n \colon |\widehat{\mP}_{x^n}(x)-\mP_X(x)| \leq \epsilon \mP_X(x), \forall x\in \mathcal{X} \big\}.
\end{align}
For jointly distributed random variables $(X_1,\ldots, X_m)$, the jointly typical set is defined by treating $(X_1,\ldots, X_m)$ as a single random variable.

\section{Problem formulation}
\subsection{Non-signaling assisted $K$-user discrete memoryless broadcast channel}
A $K$-user discrete memoryless broadcast channel (DM-BC) is specified by a conditional pmf $\mN_{Y_1Y_2\cdots Y_K\mid X} \in \mathcal{P}(\mathcal{Y}_1 \times \mathcal{Y}_2 \times \cdots \times \mathcal{Y}_K \mid \mathcal{X})$ with underlying finite-cardinality discrete sets $\mathcal{Y}_1, \mathcal{Y}_2,\ldots, \mathcal{Y}_K, \mathcal{X}$. For each use of the channel, the transmitter chooses an input $x\in \mathcal{X}$, and the probability that, for each $k\in [K]$, User $k$ receives $Y_k=y_k\in \mathcal{Y}_k$ conditioned on the channel input $X=x\in \mathcal{X}$, is equal to $\mN_{Y_1Y_2\cdots Y_K\mid X}(y_1,y_2,\ldots, y_K\mid x)$.
\begin{figure}[htbp]
\center
\begin{tikzpicture}[scale=0.75, transform shape]
	\node (Z) [
		rectangle,
		draw,
		thick,
		minimum width=13cm,
		minimum height=2cm
	] at (0,0) {};

	\draw[thick] (-1.5,1) -- (-1.5,-1);
	\draw[thick] (0.5,1) -- (0.5,-1);
	\draw[thick] (2.5,1) -- (2.5,-1);
	\draw[thick] (4.5,1) -- (4.5,-1);

	\node [
		rectangle,
		minimum height=0.7cm,
		fill=gray!30
	] at (0,0) {
		\small
		$\mZ
		\big(
			x^n,\widehat{w}_1,\widehat{w}_2,\ldots,\widehat{w}_K
			\mid
			(w_1,w_2,\ldots,w_K),
			y_1^n,y_2^n,\ldots,y_K^n
		\big)$
	};

	\node (IT) at (-4.5,0.9) {};
	\node[below=-0.25cm of IT]
		{\small $(w_1,w_2,\ldots,w_K)$};

	\node (IR1) [rectangle, minimum width=0.6cm] at (-0.5,-0.9) {};
	\node[above=-0.3cm of IR1] {\small $y_1^n$};

	\node (IR2) [rectangle, minimum width=0.6cm] at (1.5,-0.9) {};
	\node[above=-0.3cm of IR2] {\small $y_2^n$};

	\node (IRK) [rectangle, minimum width=0.6cm] at (5.5,-0.9) {};
	\node[above=-0.3cm of IRK] {\small $y_K^n$};

	\node (OT) at ($(IT.south)+(0,-1.55)$) {\small $x^n$};

	\node (OR1) at ($(IR1.north)+(0,1.5)$) {\small $\widehat{w}_1$};
	\node (OR2) at ($(IR2.north)+(0,1.5)$) {\small $\widehat{w}_2$};
	\node (ORK) at ($(IRK.north)+(0,1.5)$) {\small $\widehat{w}_K$};

	\node (N) [
		rectangle,
		draw,
		thick,
		minimum width=3cm,
		minimum height=2cm
	] at (-2.5,-3.5) {};

	\node [
		rectangle,
		minimum height=1.4cm,
		fill=gray!30
	] at (-2.5,-3.5) {
		\small
		$\mathsf{N}^{\otimes n}_{Y_1Y_2\cdots Y_K\mid X}$
	};

	\node (X) at ($(N.west)+(0.1,0)$) {};

	\node (Y1) at ($(N.east)+(-0.15,0.7)$) {};
	\node (Y2) at ($(N.east)+(-0.15,0.25)$) {};
	\node (YD)  at ($(N.east)+(-0.15,-0.2)$) {};
	\node (YK)  at ($(N.east)+(-0.15,-0.7)$) {};

	\node (Tx) [above=1cm of IT] {\sc Transmitter};

	\node (Rx1) [above=1cm of OR1] {\sc User $1$};
	\node (Rx2) [above=1cm of OR2] {\sc User $2$};
	\node at (3.5,2.3) {$\cdots$};
	\node (RxK) [above=1cm of ORK] {\sc User $K$};

	\node (C) [below=0cm of N] {\small\sc Broadcast channel};

	\draw[-{Latex[length=2.5mm]},thick]
		(Tx) -- (IT)
		node[pos=0.5,left] {$(W_1,W_2,\ldots,W_K)$};

	\draw[-{Latex[length=2.5mm]},thick]
		(OT) |- (X)
		node[pos=0.25,left] {$X^n$};

	\draw[-{Latex[length=2.5mm]},thick]
		(Y1) -| (IR1)
		node[pos=0.66,right] {$Y_1^n$};

	\draw[-{Latex[length=2.5mm]},thick]
		(Y2) -| (IR2)
		node[pos=0.72,right] {$Y_2^n$};

	\node[right=0.1cm of YD] {\small $\vdots$};

	\draw[-{Latex[length=2.5mm]},thick]
		(YK) -| (IRK)
		node[pos=0.80,right] {$Y_K^n$};

	\draw[-{Latex[length=2.5mm]},thick]
		(OR1) -- (Rx1)
		node[pos=0.5,right] {$\widehat{W}_1$};

	\draw[-{Latex[length=2.5mm]},thick]
		(OR2) -- (Rx2)
		node[pos=0.5,right] {$\widehat{W}_2$};

	\draw[-{Latex[length=2.5mm]},thick]
		(ORK) -- (RxK)
		node[pos=0.5,right] {$\widehat{W}_K$};

\end{tikzpicture}
\caption{NS-assisted coding scheme $\mZ$ for the $K$-user broadcast channel.} \label{fig:BC_K}
\end{figure}
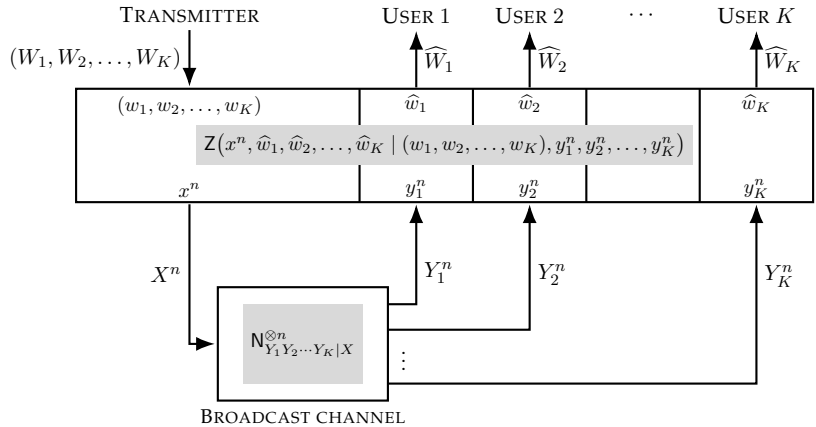

Let $M_1,M_2,\ldots, M_K, n$ be positive integers. 
An $(M_1,M_2,\ldots, M_K,n)$ NS-assisted coding scheme operating on $n$ channel uses is specified by a conditional pmf $\mZ\in \mathcal{P}(\mathcal{X}^n \times [M_1]\times [M_2] \times \cdots \times [M_K] \mid [M_1]\times [M_2] \times \cdots \times [M_K] \times \mathcal{Y}_1^n \times \mathcal{Y}_2^n \times \cdots \times \mathcal{Y}_K^n)$. The conditional pmf $\mZ$ should satisfy the $(K+1)$-partite NS conditions,
\begin{enumerate}[label=\textit{C\arabic*}:, start=0]
	\item $\sum_{x^n}\mZ(x^n,\widehat{w}_1,\widehat{w}_2,\ldots,\widehat{w}_K\mid (w_1,w_2,\ldots,w_K), y_1^n,y_2^n,\ldots,y_K^n)$ must be invariant under changes of $(w_1,w_2,\ldots, w_K)$,
\end{enumerate}
and for each $k\in [K]$,
\begin{enumerate}[label=\textit{Ck}:]
	\item $\sum_{\widehat{w}_k}\mZ(x^n,\widehat{w}_1,\widehat{w}_2,\ldots,\widehat{w}_K\mid (w_1,w_2,\ldots,w_K), y_1^n,y_2^n,\ldots,y_K^n)$ must be invariant under changes of $y_k^n$.
\end{enumerate}

Such a scheme $\mZ$ is illustrated in Fig. \ref{fig:BC_K}, represented as a black box shared across all parties where each party has its own input into and output out of the box $\mZ$, and works as follows. 
The transmitter has $K$ independent messages $(W_1,W_2,\ldots, W_K)$, where for each $k\in [K]$, $W_k$ is drawn uniformly from $[M_k]$ and is intended for User $k$. 
At the transmitter, the input to the scheme is $(W_1,W_2,\ldots, W_K)$. The scheme produces $X^n$ according to the conditional pmf $\mZ(x^n\mid (w_1,w_2,\ldots, w_K))$. 
Note that the conditional pmf is defined regardless of the values of $(y_1^n, y_2^n, \ldots, y_K^n)$, because \textit{C1}--\textit{CK} together imply that the conditional marginal pmf of $\mZ$ for $X^n$, which is obtained by summing  over all $(\widehat{w}_1, \widehat{w}_2,\ldots, \widehat{w}_K)$ in $\mZ(x^n,\widehat{w}_1,\widehat{w}_2,\ldots,\widehat{w}_K\mid (w_1,w_2,\ldots,w_K), y_1^n,y_2^n,\ldots,y_K^n)$, does not depend on $(y_1^n, y_2^n, \ldots, y_K^n)$. 
The transmitter then sends $X^n$ through $n$ uses of the broadcast channel. For each $k\in [K]$, User $k$ receives $Y_k^n$ from the channel and inputs $Y_k^n$ to the scheme. 
The scheme produces $(\widehat{W}_1,\widehat{W}_2,\ldots, \widehat{W}_K)$, where $\widehat{W}_k$ is produced at User $k$, according to the conditional joint pmf $\mZ(\widehat{w}_1, \widehat{w}_2, \ldots, \widehat{w}_K \mid (w_1,w_2,\ldots, w_K), x^n, y_1^n, y_2^n, \ldots, y_K^n)$.

The joint distribution for the random variables $(W_1,\ldots, W_K, X^n, Y_1^n,  \ldots, Y_K^n, \widehat{W}_1,  \ldots, \widehat{W}_K)$ is (cf. \cite[Eq. (9)]{Yao_Jafar_NS_DoF}),
\begin{align}
	&\Pr(W_k=w_k, X^n=x^n, Y_k^n=y_k^n, \widehat{W}_k = \widehat{w}_k, \forall k\in [K]) \notag \\
	&=\big( \prod_{k=1}^K M_k \big)^{-1}\prod_{i=1}^n \mN_{Y_1Y_2\cdots Y_K\mid X}(y_{1,i},y_{2,i},\ldots, y_{K,i}\mid x_i) \notag \\
	&\hspace{1cm} \times \mZ(x^n,\widehat{w}_1,\ldots, \widehat{w}_K\mid (w_1,w_2,\ldots, w_K), y_1^n, y_2^n, \ldots, y_K^n)
\end{align}
For each $k\in [K]$, define the probability of decoding error for User $k$ as,
\begin{align}
P_{e,k}(\mZ) = \Pr(\widehat{W}_k\neq W_k).
\end{align}
 
\subsection{Achievable rate tuples and capacity region}
Given the channel $\mN_{Y_1Y_2\cdots Y_K \mid X}$, a rate tuple $(R_1,R_2,\ldots, R_K)\in \mathbb{R}_{\geq 0}^K$ is said to be achievable by NS-assisted coding schemes if there exists a sequence of $(M_{1}^{(n)}, \ldots, M_{K}^{(n)},n)$ NS-assisted coding schemes $(\mZ_n)_{n\in \mathbb{N}}$ such that 
\begin{align}
\lim_{n\to \infty} P_{e,k} (\mZ_n) &= 0,\\
\mbox{and }\liminf_{n\to \infty} \frac{\log_2 M_{k}^{(n)}}{n} &\geq R_k, \quad\forall k\in [K].
\end{align}
The NS-assisted capacity region $\mathcal{C}^{\NS}(\mN_{Y_1Y_2\ldots Y_K\mid X})$ is defined as the closure of all rate tuples achievable by NS-assisted coding schemes.

\begin{remark}[Same-marginals property] \label{rem:same_marginals}
	Similar to a known result for classical coding schemes of broadcast channels, it is shown in \cite[Thm. 2]{Yao_Jafar_NS_DoF} that for NS-assisted coding schemes, the probability of decoding error for User $k$, i.e., $P_{e,k}(\mZ)$ depends only on $\mZ$ and $\mN_{Y_k\mid X}$, i.e., the channel's conditional marginal pmf for User $k$. It follows that $\mathcal{C}^{\NS}(\mN_{Y_1Y_2\cdots Y_K\mid X})$ is determined by the $K$ conditional marginal pmfs $(\mN_{Y_1\mid X}, \mN_{Y_2\mid X}, \ldots, \mN_{Y_K\mid X})$.
\end{remark}

\section{Results}
\subsection{NS-assisted capacity region of $K$-user DM-BC}\label{sec:NSBC_Kuser}
The main result of this paper is stated in the following theorem.
\begin{theorem}[NS-assisted $K$-user BC capacity region] \label{thm:NSBC_Kuser}
	For a $K$-user discrete memoryless broadcast channel $\mN_{Y_1Y_2\cdots Y_K\mid X}$, the NS-assisted capacity region is
	\begin{align}
		\mathcal{C}^{\NS}(\mN_{Y_1Y_2\cdots Y_K\mid X}) = \mathcal{R}_{\Sato}(\mN_{Y_1Y_2\cdots Y_K\mid X}),
	\end{align}
	where the `Sato' region is defined as,
	\begin{equation}
	\begin{aligned}
		&\mathcal{R}_{\Sato}(\mN_{Y_1Y_2\cdots Y_K\mid X}) \\
		& \triangleq\bigcup_{\mP_X\in \mathcal{P}(\mathcal{X})} \left\{
	  	\begin{array}{l}
	    	(R_1,R_2,\ldots, R_K) \in \mathbb{R}_{\geq 0}^K \colon \\
	    	\sum_{k\in \mathcal{K}} R_k \leq J(\mP_X, (\mN_{Y_{\ell}\mid X})_{\ell \in \mathcal{K}}),~~\forall \emptyset \neq  \mathcal{K} \subseteq [K] 
	  \end{array}
	  \right.
	\end{aligned}
	\end{equation}
	and $J(\mP_X, (\mN_{Y_{\ell}\mid X})_{\ell \in \mathcal{K}})$ is defined in \eqref{eq:def_JK}.
\end{theorem}

\noindent The converse, $\mathcal{C}^{\NS} \subseteq \mathcal{R}_{\Sato}$, for Theorem \ref{thm:NSBC_Kuser}  is presented in Appendix \ref{proof:converse}, and is relatively straightforward, as it follows from a generalization of the two-user converse \cite[Thm. 4]{Yao_Jafar_CSITTP}. Intuitively, for each $\mP_X$, a data-processing argument and the same-marginals property, combined with the facts that cooperation among receivers cannot hurt, and that NS-assistance cannot increase the capacity of the point-to-point discrete memoryless channel \cite{matthews2012linear} obtained by receiver-cooperation, lead to the converse bounds.

The proof of achievability, $\mathcal{R}_{\Sato} \subseteq \mathcal{C}^{\NS}$, is the main challenging part of the theorem.  The key ingredient for this proof is the `authentication solution' previously developed for the proof of the NS-assisted capacity of a point-to-point channel with causal CSIT in \cite{Yao_Jafar_NSCWS} (see also \cite{OTPmodel}). Intuitively, the authentication solution $\mZ$ strives to output the correct desired message $\widehat{w}_k=w_k$ at each User $k$, $k\in[K]$, if that user's input $y_k^n$ to the box $\mZ$ passes a certain authentication test, namely that it is jointly typical with the output $x^n$ that $\mZ$ produced at the transmitter to be sent over the channel, with typicality defined according to the marginal distribution $\mP_{X Y_k}=\mP_X \mN_{Y_k|X}$ induced by the channel between the transmitter and that user. The delicate part of the construction is to balance the probability values so that the NS conditions are satisfied \emph{exactly}. To convey the main ideas in a simpler setting, we first provide the achievability proof for the two-user $(K=2)$ BC in Section \ref{proof:NSBC_capacity}. The proof for the $K$-user BC is relegated to Appendix \ref{proof:NSBC_Kuser}.

\subsection{Resolving an open question in \cite{Yao_Jafar_CSITTP}}
Specializing Theorem \ref{thm:NSBC_Kuser} to $K=2$, the NS-assisted capacity region of a two-user DM-BC $\mN_{Y_1Y_2\mid X}$ can be stated explicitly as,
\begin{align}
\mathcal{C}^{\NS}(\mN_{Y_1Y_2\mid X})= \mathcal{R}_{\Sato}(\mN_{Y_1Y_2\mid X}) = \bigcup_{\mP_X \in \mathcal{P}(\mathcal{X})} \left\{
  \begin{array}{l}
    (R_1,R_2)\in \mathbb{R}_{\geq 0}^{2}\colon \\
    R_1 \leq I_{\mP_X \mN_{Y_1\mid X}}(X;Y_1) \\
    R_2 \leq I_{\mP_X \mN_{Y_2\mid X}}(X;Y_2) \\
    R_1 + R_2 \leq J(\mP_X, (\mN_{Y_1\mid X}, \mN_{Y_2\mid X}))
  \end{array}
  \right.
\end{align}
Observe that $\mathcal{R}_{\Sato}(\mN_{Y_1Y_2\mid X}) $ is precisely the Sato's region for the $2$-user DM-BC (cf. \cite[Eq. (25)]{Yao_Jafar_CSITTP}, \cite[Problem 8.8]{NIT}). It is shown in \cite[Thm. 4]{Yao_Jafar_CSITTP} that $\mathcal{C}^{\NS}(\mN_{Y_1Y_2\mid X}) \subseteq \mathcal{R}_{\Sato}(\mN_{Y_1Y_2\mid X})$, and it was left as an open question whether $\mathcal{R}_{\Sato}(\mN_{Y_1Y_2\mid X})$ is achievable by NS-assisted coding schemes for a general two-user DM-BC.  Theorem \ref{thm:NSBC_Kuser} thus answers this question in the affirmative.

\subsection{A key lemma} \label{sec:preliminaries}
The following lemma represents a key technical argument for the achievability proof of Theorem \ref{thm:NSBC_Kuser}, as it essentially bounds the probability that arbitrary inputs applied to $\mZ$ by a subset of users are able to simultaneously pass the authentication test.
\begin{lemma}[Simultaneous typicality upper bound] \label{lem:typical_off_K}
	Let $\kappa\in \mathbb{N}, n\in \mathbb{N}, \epsilon \in (0,1)$, and $\mP_X\in \mathcal{P}(\mathcal{X})$. For each $k\in [\kappa]$, let $\mP_{Y_k\mid X}\in \mathcal{P}(\mathcal{Y}_k\mid \mathcal{X})$ and let $y_k^n \in \mathcal{Y}_k^n$ be arbitrary. Let $X^n\sim \mP_X^{\otimes n}$.  Then there exists $\delta(\epsilon) \geq 0$ that tends to $0$ as $\epsilon \to 0$ such that $\Pr\big((X^n, y_k^n) \in \mathcal{T}^{(n)}_{\epsilon}(\mP_X \mP_{Y_k\mid X}), \forall k\in [\kappa] \big) \leq 2^{-n (J(\mP_X, (\mP_{Y_k\mid X})_{k\in [\kappa]})-\delta(\epsilon))}$, where $J(\mP_X, (\mP_{Y_k\mid X})_{k\in [\kappa]})$ is defined in \eqref{eq:def_JK}.
\end{lemma}
\begin{proof}
	See Appendix \ref{proof:typical_off_K}.
\end{proof}

For  $\kappa =1$, we have $J(\mP_X, \mP_{Y_1\mid X}) = I_{\mP_X \mP_{Y_1\mid X}}$, and Lemma \ref{lem:typical_off_K} reduces to a special case of the joint typicality lemma \cite[Sec. 2.5.1]{NIT}. To intuitively see why Lemma \ref{lem:typical_off_K} holds in general, let us provide  a proof sketch here, leaving the formal proof to Appendix \ref{proof:typical_off_K}. Fix arbitrary sequences $y_1^n, y_2^n, \ldots, y_{\kappa}^n$, and let $\mathcal{T}\triangleq \{ x^n\colon (x^n, y_k^n) \in \mathcal{T}_{\epsilon}^{(n)}(\mP_X \mP_{Y_k\mid X}), \forall k\in [\kappa]\}$. For the non-trivial case $\mathcal{T}\neq \emptyset$, we partition the sequences $x^n \in \mathcal{T}$ according to the joint type with $(y_1^n,\ldots, y_{\kappa}^n)$, i.e., $\widehat{\mP}_{x^n y_1^n \cdots y_{\kappa}^n}$, so that $x^n$ belongs to the class $\mQ_{XY_1\cdots Y_{\kappa}}$ if $\widehat{\mP}_{x^n y_1^n \cdots y_{\kappa}^n} = \mQ_{XY_1\cdots Y_{\kappa}}$.
It follows by definition that each class $\mQ_{XY_1\cdots Y_{\kappa}} \in \mathcal{Q}_{\epsilon}$, where $\mathcal{Q}_{\epsilon} \triangleq \{\mQ_{XY_1\cdots Y_{\kappa}}\colon |\mQ_{XY_k}(x,y_k) - \mP_{XY_k}(x,y_k)| \leq \epsilon \mP_{XY_k}(x,y_k), \forall k\in [\kappa]\}$. Now consider the class of $x^n$ for which the joint type is $\mQ_{XY_1\cdots Y_{\kappa}}$. The total probability of $X^n\sim \mP_{X}^{\otimes n}$ being in this class is at most $2^{-n I_{\mQ}(X;Y_1,\ldots, Y_{\kappa})}$. Since the number of possible joint types is polynomial in $n$, by summing over all joint types in $\mathcal{Q}_{\epsilon}$, we have $\Pr(X^n \in \mathcal{T}) \leq {\rm poly}(n) \max_{\mQ \in \mathcal{Q}_{\epsilon}}2^{-n I_{\mQ}(X;Y_1,\ldots, Y_{\kappa})}$, where ${\rm poly}(n)$ denotes a function whose growth is at most polynomial in $n$. Equivalently, $\Pr(X^n \in \mathcal{T}) \leq 2^{-n J_{\epsilon} + o(n)}$ where $J_{\epsilon} \triangleq \min_{\mQ \in \mathcal{Q}_{\epsilon}} I_{\mQ}(X;Y_1,\ldots, Y_{\kappa})$. By a continuity argument, one can show that $\Pr(X^n \in \mathcal{T}) \leq 2^{-n (J_0 - \delta(\epsilon)) + o(n)} = 2^{-n (J(\mP_X, (\mP_{Y_k\mid X})_{k\in [\kappa]}) - \delta(\epsilon)) + o(n)}$.  Remarkably, the formal proof is sharper than the argument sketched above, in that the extra $o(n)$ term in the exponent is shown to be unnecessary. Indeed, the bound in Lemma \ref{lem:typical_off_K} does not require $n$ to be sufficiently large; rather, it holds for \emph{all} $n$.

\subsection{Comparison with bipartite NS assistance}
Consider a restricted setting where only bipartite non-signaling resources are available.
 For a two-user broadcast channel, a bipartite non-signaling resource may be shared between the pairs of parties (Transmitter, User 1), (Transmitter, User 2), and (User 1, User 2). It is known from \cite{pereg2021quantumBC, fawzi2024broadcast} that non-signaling assistance, when established only between the two users, does not improve the capacity region. In addition, it is proved in \cite{Yao_Jafar_CSITTP} that when the transmitter has non-signaling assistance only with User $1$, the capacity region is described by the following region, referred to as the Kramer-Shamai (KS) region \cite{kramer2007capacity} with respect to User 1,
\begin{align}
  \mathcal{R}_{\KS,1}(\mN_{Y_1Y_2\mid X}) = \bigcup_{\mP_{XU} \in \mathcal{P}(\mathcal{X} \times \mathcal{U})} \left\{
  \begin{array}{l}
    (R_1,R_2)\in \mathbb{R}_{\geq 0}^{2}\colon \\
    R_1 \leq I(X;Y_1) \\
    R_2 \leq I(U;Y_2) \\
    R_1 + R_2 \leq I(X;Y_1\mid U) + I(U;Y_2)
  \end{array}
  \right.
\end{align}
where  $|\mathcal{U}|\leq |\mathcal{X}|+1$ and $U \leftrightarrow X \leftrightarrow (Y_1,Y_2)$ form a Markov chain. It corresponds to the classical capacity region when $W_2$ is provided to User $1$.
By symmetry, the capacity region when the transmitter has non-signaling assistance only with User $2$ is $\mathcal{R}_{\KS,2}(\mN_{Y_1Y_2\mid X})$, obtained by exchanging the user indices in $\mathcal{R}_{\KS,1}(\mN_{Y_1Y_2\mid X})$. A natural extension of these `\emph{bipartite-NS}'-assisted capacity regions is the following (KS+) region
\begin{align}
	\mathcal{R}_{\KSp} \triangleq \conv\big(\mathcal{R}_{\KS,1} \cup \mathcal{R}_{\KS,2} \big)
\end{align}
where $\conv(\cdot)$ denotes the convex hull. 

It is shown by \cite{Yao_Jafar_CSITTP} that for some special classes of two-user BCs, $\mathcal{R}_{\KSp}$ matches $\mathcal{R}_{\Sato}$ so they are both equal to $ \mathcal{C}^{\NS}$. Such examples include the Erasure Blackwell channel $\mN_{\BLEC}$ \cite{gohari2021outer}, for which $\mathcal{R}_{\KS,1}(\mN_{\BLEC}) = \mathcal{R}_{\Sato}(\mN_{\BLEC})$.
Since in general $\mathcal{R}_{\KSp} \subseteq \mathcal{C}^{\NS}$, one might ask if the inclusion can be strict, i.e., whether for some channels we have $\mathcal{R}_{\KSp} \subset \mathcal{C}^{\NS}$. The following corollary (Corollary \ref{cor:BSSC}) confirms the existence of such channels. The example used in Corollary \ref{cor:BSSC} is the binary skew-symmetric channel (BSSC) $\mN_{\BSSC}$. The channel alphabets are $\mathcal{X}=\mathcal{Y}_1=\mathcal{Y}_2=\{0,1\}$. The channel is defined as follows: if $X=0$, then $Y_1=0$ and $Y_2$ is uniformly distributed over $\{0,1\}$; if $X=1$, then $Y_2=1$ and $Y_1$ is uniformly distributed over $\{0,1\}$.
\begin{corollary} \label{cor:BSSC}
	For the two-user DM-BC represented by the binary skew-symmetric channel (BSSC) $\mN_{\BSSC}$, we have,
	\begin{align}
	\max_{(R_1,R_2)\in \mathcal{C}^{\NS}(\mN_{\BSSC})} (R_1+R_2) &\geq 0.5,\\ 
	\max_{(R_1,R_2)\in \mathcal{R}_{\KSp}(\mN_{\BSSC})} (R_1+R_2) &< 0.45.
	\end{align}
	Thus,  $\mathcal{R}_{\KSp}(\mN_{\BSSC}) \subset \mathcal{C}^{\NS}(\mN_{\BSSC})$, i.e., the inclusion is strict.
\end{corollary}
\noindent We provide the proof in Appendix \ref{proof:BSSC}.

\section{Achievability of Theorem \ref{thm:NSBC_Kuser}: Two-user case} \label{proof:NSBC_capacity}
In the following, we write $I_{\mP_X\mN_{Y_k\mid X}}(X;Y_k)$ simply as $I(X;Y_k)$ for each $k\in \{1,2\}$, and write $J(\mP_X, (\mN_{Y_1\mid X}, \mN_{Y_2\mid X}))$ as $J$.
We shall prove that, for each $\mP_X\in \mathcal{P}(\mathcal{X})$, $(R_1,R_2)\in \mathbb{R}_{> 0}^2$ is achievable if
\begin{align}
	R_1 &< I  (X;Y_1)\\
	R_2 &< I  (X;Y_2)\\
	R_1+R_2 &< J 
\end{align}
whenever $\min_{k\in \{1,2\}} \{I  (X;Y_k)\} > 0$. Note that if, for example, $I(X;Y_1)=0$, then the achievability of any $R_2<I(X;Y_2)$ is implied by the degenerate single-user solution with User $2$.

\subsection{NS-assisted coding schemes for $K=2$}
An $(M_1,M_2,n)$ NS-assisted coding scheme over $n$ uses of the channel is specified by a conditional pmf $\mZ \in \mathcal{P}(\mathcal{X}^n \times [M_1]\times [M_2] \mid [M_1]\times [M_2] \times \mathcal{Y}_1^n \times \mathcal{Y}_2^n)$, for which the following set of conditions (tripartite non-signaling conditions) hold.
\begin{enumerate}[label=\textit{C\arabic*}:, start=0]
	\item $\sum_{x^n}\mZ(x^n,\widehat{w}_1, \widehat{w}_2\mid (w_1,w_2), y_1^n, y_2^n)$ is invariant under changes of $(w_1,w_2)\in [M_1]\times [M_2]$;
	\item $\sum_{\widehat{w}_1}\mZ(x^n,\widehat{w}_1, \widehat{w}_2\mid (w_1,w_2), y_1^n, y_2^n)$ is invariant under changes of $y_1^n\in \mathcal{Y}_1^n$;
	\item $\sum_{\widehat{w}_2}\mZ(x^n,\widehat{w}_1, \widehat{w}_2\mid (w_1,w_2), y_1^n, y_2^n)$ is invariant under changes of $y_2^n\in \mathcal{Y}_2^n$.
\end{enumerate}

\subsection{Towards an authentication solution: Key properties}
An NS-assisted coding scheme $\mZ(x^n, \widehat{w}_1, \widehat{w}_2\mid (w_1,w_2), y_1^n, y_2^n)$ admits the following factorization
\begin{align}
	&\mZ(x^n, \widehat{w}_1, \widehat{w}_2\mid (w_1,w_2), y_1^n, y_2^n)\notag\\
	&=\mZ(x^n \mid (w_1,w_2)) \mZ(\widehat{w}_1, \widehat{w}_2 \mid (w_1,w_2), x^n, y_1^n, y_2^n).
\end{align}
Observe that the first factor does not depend on $(y_1^n,y_2^n)$ due to the NS conditions \textit{C1} and \textit{C2}.

We will focus on a subclass of NS-assisted coding schemes, referred to as the \emph{authentication solution}, which has the following properties.
\begin{enumerate}[label=P\arabic*:]
	\item At the transmitter's side, the scheme generates $X^n$ according to a distribution that does not depend on $(w_1,w_2)$. Formally, 
	\begin{align}\mZ(x^n \mid (w_1,w_2)) = \mZ(x^n)
	\end{align}
	 for some distribution $\mZ(x^n)$ that is invariant under changes of $(w_1,w_2)$.
	\item At the receivers' side, according to some distribution that only depends on $(x^n,y_1^n,y_2^n)$, for each user, the scheme either outputs the correct message, or an incorrect message chosen uniformly at random for that user. Specifically, for $k\in \{1,2\}$, let $a_k$ be the probability of the event that ``the scheme outputs the correct message for User $k$." Similarly, let $a_{12}$ be the probability of the event that ``the scheme outputs the correct messages for both users." Formally, 
	\begin{align}
		&\mZ(\widehat{w}_1, \widehat{w}_2 \mid (w_1,w_2), x^n, y_1^n, y_2^n) \notag\\
		&=\begin{cases}
			 a_{12}, & \mbox{if}~ \widehat{w}_1=w_1, \widehat{w}_2=w_2 \\
			 \frac{a_1-a_{12}}{M_2-1}, & \mbox{if}~ \widehat{w}_1=w_1, \widehat{w}_2\neq w_2\\
			 \frac{a_2-a_{12}}{M_1-1}, & \mbox{if}~ \widehat{w}_1\neq w_1, \widehat{w}_2=w_2\\
			 \frac{1-a_1-a_2+a_{12}}{(M_1-1)(M_2-1)}, & \mbox{if}~ \widehat{w}_1\neq w_1, \widehat{w}_2 \neq w_2
		\end{cases}.
	\end{align}
	The values $a_1,a_2,a_{12}$ depend on the parameters $(x^n,y_1^n,y_2^n)$, but the dependence is not explicitly indicated for the sake of compact notation. Observe that $a_{1}-a_{12}$ equals the probability of the event that ``the scheme outputs the correct message for User 1 and an incorrect message for User 2." Similarly, $a_{2}-a_{12}$ equals the probability of the event that ``the scheme outputs the correct message for User 2 and an incorrect message for User 1," and $1-a_{1}-a_{2}+a_{12}$ equals the probability of the event that ``the scheme outputs incorrect messages for both users." Since a probability value must be in $[0,1]$, $a_1,a_2,a_{12}$ must satisfy the condition
	\begin{align} \max\{0, a_1+a_2-1\}\leq a_{12}\leq \min \{a_1,a_2\}.\label{eq:validprob}
	\end{align}
	Conversely, any choice satisfying Condition \eqref{eq:validprob} yields valid probability values.
	Now consider the NS conditions. For the NS condition \textit{C2} to hold, it suffices that $a_1$ must be invariant under changes of $y_2^n$. In other words, the probability of outputting the correct message for User $1$ must not depend on the input at User $2$. We have,
\begin{align}
	&\sum_{\widehat{w}_2} \mZ(x^n, \widehat{w}_1, \widehat{w}_2\mid (w_1,w_2), y_1^n, y_2^n)\notag \\
	&=\mZ(x^n) \sum_{\widehat{w}_2} \mZ(\widehat{w}_1, \widehat{w}_2 \mid (w_1,w_2), x^n, y_1^n, y_2^n)  \\
	&= \mZ(x^n) \begin{cases}
		a_1, & \mbox{if}~ \widehat{w}_1=w_1\\
		\frac{1-a_1}{M_1-1}, & \mbox{if}~ \widehat{w}_1\neq w_1.
	\end{cases}
\end{align}
Thus, \textit{C2} holds if $a_1$ depends only on $(x^n,y_1^n)$.
Similarly, \textit{C1}  is satisfied if $a_2$ depends only on $(x^n,y_2^n)$.
\end{enumerate}

In addition, for \textit{C0} to hold, it suffices to have 
\begin{enumerate}[label=P\arabic*:, start=3]
	\item \begin{align}
	&\mathbb{E}[a_{1}(X^n, y_1^n)] = \frac{1}{M_1}, \label{eq:P3_1}\\
	&\mathbb{E}[a_{2}(X^n, y_2^n)] = \frac{1}{M_2}, \label{eq:P3_2}\\
	&\mathbb{E}[a_{12}(X^n, y_1^n, y_2^n)] = \frac{1}{M_1M_2}, \label{eq:P3_12}
\end{align}
where the expectations are defined with respect to $X^n \sim \mZ(x^n)$. 
To see this, note that, by tracing out $x^n$,
\begin{align}
	&\sum_{x^n}\mZ(x^n, \widehat{w}_1, \widehat{w}_2\mid (w_1,w_2), y_1^n, y_2^n)\notag \\
	&=\begin{cases}
			\mathbb{E}[a_{12}(X^n, y_1^n, y_2^n)], & \mbox{if}~ \widehat{w}_1=w_1, \widehat{w}_2=w_2 \\
			 \frac{\mathbb{E}[a_1(X^n,y_1^n)]-\mathbb{E}[a_{12}(X^n, y_1^n, y_2^n)]}{M_2-1}, & \mbox{if}~ \widehat{w}_1=w_1, \widehat{w}_2\neq w_2\\
			 \frac{\mathbb{E}[a_2(X^n,y_2^n)]-\mathbb{E}[a_{12}(X^n, y_1^n, y_2^n)]}{M_1-1}, & \mbox{if}~ \widehat{w}_1\neq w_1, \widehat{w}_2=w_2\\
			 \frac{1-\mathbb{E}[a_1(X^n,y_1^n)]-\mathbb{E}[a_2(X^n,y_2^n)]+\mathbb{E}[a_{12}(X^n,y_1^n,y_2^n)]}{(M_1-1)(M_2-1)}, & \mbox{if}~ \widehat{w}_1\neq w_1, \widehat{w}_2 \neq w_2
		\end{cases}\\
	&= \frac{1}{M_1M_2}
\end{align}
which is invariant under changes of $(w_1,w_2)$.
\end{enumerate}

\begin{remark}[Twirling]
	Despite the apparent specialization to enforce the properties P1--P3, there is no loss of generality in terms of the optimal probability of error by considering only the authentication solutions.  This can be shown by a twirling argument similar to the one used in \cite{cubitt2011zero}. Specifically, let $\mZ$ be any NS-assisted coding scheme. Define 
	\begin{align}
	&\mZ'(x^n,\widehat{w}_1, \widehat{w}_2\mid (w_1,w_2), y_1^n, y_2^n) \notag\\
	&= \frac{1}{|\mathcal{S}_{M_1}||\mathcal{S}_{M_2}|}\sum_{\pi_1\in \mathcal{S}_{M_1}, \pi_2 \in \mathcal{S}_{M_2}} \mZ(x^n, \pi_1(\widehat{w}_1), \pi_2(\widehat{w}_2) \mid (\pi_1(w_1), \pi_2(w_2)), y_1^n, y_2^n),
	\end{align} where $\mathcal{S}_m$ denotes the set of all permutations of the set $[m]$. It can be verified that $\mZ'$ is also a NS-assisted coding scheme, and that $\mZ'$ satisfies P1--P3 due to the symmetry in $\mZ'$. Moreover, $\mZ'$ has the same  probability of error as $\mZ$, i.e., $P_{e,k}(\mZ) = P_{e,k}(\mZ'), \forall k\in \{1,2\}$.
\end{remark}

\subsection{Authentication solution}
Recall that we have fixed $\mP_X$. Let $\epsilon \in (0,1)$.
Define the following probabilities,
\begin{align}
	& t_1(y_1^n) \triangleq \Pr\big( (X^n,y_1^n) \in \mathcal{T}^{(n)}_{\epsilon}(\mP_X \mN_{Y_1\mid X}) \big)\\
	& t_2(y_2^n) \triangleq \Pr\big( (X^n,y_2^n) \in \mathcal{T}^{(n)}_{\epsilon}(\mP_X \mN_{Y_2\mid X}) \big)\\
	&t_{12}(y_1^n,y_2^n) \triangleq \Pr\Big( (X^n,y_1^n) \in \mathcal{T}^{(n)}_{\epsilon}(\mP_X \mN_{Y_1\mid X}) \notag \\
	& \hspace{2.5cm}\mbox{and}~ (X^n,y_2^n) \in \mathcal{T}^{(n)}_{\epsilon}(\mP_X \mN_{Y_2\mid X}) \Big)
\end{align}
where $X^n \sim \mP_{X}^{\otimes n}$.

We next specify the authentication solution. Firstly, let
\begin{align}
	\mZ(x^n) \triangleq \mP_X^{\otimes n}(x^n).
\end{align}
This ensures that P1 is satisfied.

Then, for each $(x^n, y_1^n)$, let
\begin{align}\label{eq:choosea1}
	a_{1}(x^n,y_1^n) \triangleq \begin{cases}
		1, & \mbox{if}~ (x^n,y_1^n) \in \mathcal{T}^{(n)}_{\epsilon}(\mP_X \mN_{Y_1\mid X}) \\
		c_1(y_1^n), & \mbox{otherwise}
	\end{cases}.
\end{align}
Intuitively, $a_1$ is a `typicality authenticator,' which outputs $1$ when (albeit not necessarily \emph{only} when) the pair $(x^n,y_1^n)$ is jointly typical. More importantly, observe that $a_1$ does not depend on $y_2^n$, which ensures that the NS-condition \textit{C2} is satisfied. The weights $0\leq c_1(y_1^n)\leq 1$ will be chosen later to satisfy \eqref{eq:P3_1}.

Similarly, let
\begin{align}\label{eq:choosea2}
	a_{2}(x^n,y_2^n) \triangleq \begin{cases}
		1, & \mbox{if}~ (x^n,y_2^n) \in \mathcal{T}^{(n)}_{\epsilon}(\mP_X \mN_{Y_2\mid X}) \\
		c_2(y_2^n), & \mbox{otherwise}
	\end{cases},
\end{align}
Observe that $a_2$ does not depend on $y_1^n$, which ensures that the NS-condition \textit{C1} is satisfied.
The weights $0\leq c_2(y^n) \leq 1$ will be chosen to satisfy \eqref{eq:P3_2}.

Moreover, let
\begin{align} \label{eq:def_a12}
	a_{12}(x^n,y_1^n,y_2^n) \triangleq \begin{cases}
		1, & \mbox{if}~ (x^n,y_k^n) \in \mathcal{T}^{(n)}_{\epsilon}(\mP_X \mN_{Y_k\mid X}), ~~\forall k\in \{1,2\} \\
		c_1(y_1^n), & \mbox{if} (x^n, y_1^n) \not\in \mathcal{T}^{(n)}_{\epsilon}(\mP_X \mN_{Y_1\mid X}), (x^n, y_2^n)\in \mathcal{T}^{(n)}_{\epsilon}(\mP_X \mN_{Y_2\mid X})\\
		c_2(y_2^n), & \mbox{if} (x^n, y_1^n) \in \mathcal{T}^{(n)}_{\epsilon}(\mP_X \mN_{Y_1\mid X}), (x^n, y_2^n)\not\in \mathcal{T}^{(n)}_{\epsilon}(\mP_X \mN_{Y_2\mid X})\\
		c_{12}(y_1^n, y_2^n), & \mbox{otherwise}
	\end{cases},
\end{align}
and the weights $0\leq c_{12}(y_1^n,y_2^n)\leq 1$ will be chosen to satisfy \eqref{eq:P3_12}.

Next we specify these weights to directly satisfy \eqref{eq:P3_1}, \eqref{eq:P3_2}, \eqref{eq:P3_12} and thus satisfy P3.
Note that \eqref{eq:P3_1} and \eqref{eq:P3_2} require that, for each $k\in \{1,2\}$, $\frac{1}{M_k}=\mathbb{E}[a_k(X^n, y_k^n)]= t_k(y_k^n) + (1-t_k(y_k^n))c_k(y_k^n)$. We therefore let, for each $k\in \{1,2\}$,
\begin{align}
	& c_k(y_k^n) = \frac{1/M_k - t_k(y_k^n)}{1-t_k(y_k^n)}.
\end{align}
In addition, \eqref{eq:P3_12} requires that,
\begin{align}
	&\frac{1}{M_1M_2} = \mathbb{E}[a_{12}(X^n,y_1^n,y_2^n)]\notag \\
	&=t_{12} + (t_2-t_{12})c_1 + (t_1-t_{12})c_2 + (1-t_1-t_2+t_{12})c_{12}.
\end{align}
For this we let
\begin{align}
	c_{12} = \frac{\frac{1}{M_1M_2}-t_{12}-(t_2-t_{12})c_1-(t_1-t_{12})c_2}{1-t_1-t_2+t_{12}},
\end{align}
where the parameters $y_1^n$, $y_2^n$ are suppressed, as they are clear from the context. Observe that with these choices, P3 is satisfied.

\subsection{Validity analysis}
Let us check the validity of the solution. Since P1 and P3 are satisfied directly by the construction, and the specification of $a_1$ and $a_2$ in \eqref{eq:choosea1}, \eqref{eq:choosea2} satisfies NS-conditions \textit{C2, C1},  it remains only to verify that the probability values are valid, i.e., \eqref{eq:validprob} is satisfied.

Lemma \ref{lem:typical_off_K} (setting $\kappa=1$) implies that, for each $k\in \{1,2\}$,
\begin{align}
	t_k(y_k^n) \leq 2^{-n(I (X;Y_k)-\delta_k(\epsilon))}
\end{align}
for every $y_k^n \in \mathcal{Y}_k^n$, where $\delta_k(\epsilon)>0$ and $\lim_{\epsilon \to 0+}~ \delta_k(\epsilon) = 0$. 

Moreover, Lemma \ref{lem:typical_off_K} (setting $\kappa=2$) implies that
\begin{align}
	t_{12}(y_1^n,y_2^n) \leq 2^{-n(J -\delta_{12}({\epsilon}))}
\end{align}
for every $(y_1^n,y_2^n)\in \mathcal{Y}_1^n\times \mathcal{Y}_2^n$, where $\lim_{\epsilon\to 0+}~ \delta_{12}(\epsilon) = 0$. 

From now on, restrict $\epsilon>0$ to be sufficiently small so that 
\begin{align}
I(X;Y_k) -\delta_k(\epsilon) &>0, \quad\forall k\in \{1,2\}\label{eq:epsilonsmalla}\\
\mbox{and } J -\delta_{12}(\epsilon) &>0.\label{eq:epsilonsmallb}
\end{align}
For any $(R_1,R_2)\in \mathbb{R}_{> 0}^2$ such that
\begin{equation}\label{eq:achi_rates}
\begin{aligned}
	R_1 &< I(X;Y_1) - \delta_1(\epsilon)\\
	R_2 &< I(X;Y_2) - \delta_2(\epsilon)\\
	R_1+R_2 &< J - \delta_{12}(\epsilon)
\end{aligned}
\end{equation}
by picking
\begin{align}
	&M_1^{(n)} \triangleq \lfloor 2^{nR_1} \rfloor, \qquad M_2^{(n)} \triangleq \lfloor 2^{nR_2} \rfloor,
\end{align}
it is guaranteed that, as $n\to \infty$,
\begin{align}
	&\frac{\log_2 (M_k^{(n)})}{n} \to R_k,\quad\forall k\in \{1,2\},
\end{align}
and that for every $(y_1^n,y_2^n)$,
\begin{align}
	&t_1(y_1^n) \to 0, ~~~~~~t_2(y_2^n) \to 0, ~~~~~~ t_{12}(y_1^n,y_2^n) \to 0\\
	&\frac{t_1(y_1^n)}{1/M_1^{(n)}} \to 0, ~~ \frac{t_2(y_2^n)}{1/M_2^{(n)}} \to 0, ~~ \frac{t_{12}(y_1^n,y_2^n)}{1/(M_1^{(n)}M_2^{(n)})} \to 0 \\
	&\frac{c_1(y_1^n)}{1/M_1^{(n)}} \to 1, ~~ \frac{c_2(y_2^n)}{1/M_2^{(n)}} \to 1, ~~ \frac{c_{12}(y_1^n,y_2^n)}{1/(M_1^{(n)}M_2^{(n)})} \to 1 \label{eq:limit_3}
\end{align}
because
\begin{align}
	&t_1(y_1^n) \leq 2^{-n(I(X;Y_1) -\delta_1(\epsilon))} = 2^{-n(R_1+\delta_1')} < 2^{-nR_1} \leq 1/M_1^{(n)}\\
	&t_2(y_2^n) \leq 2^{-n(I(X;Y_2) -\delta_2(\epsilon))} = 2^{-n(R_2+\delta_2')} < 2^{-nR_2} \leq 1/M_2^{(n)}\\
	&t_{12}(y_1^n,y_2^{n}) \leq 2^{-n(J -\delta_{12}(\epsilon))} = 2^{-n(R_1+R_2+\delta_{12}')} < 2^{-n(R_1+R_2)} \leq 1/(M_1^{(n)}M_2^{(n)})
\end{align}
where $\delta_1', \delta_2', \delta_{12}'>0$.
Since both $1/M_1^{(n)}$ and $1/M_2^{(n)}$ are positive sequences that tend to $0$, it follows from \eqref{eq:limit_3} that, for sufficiently large $n$, 
\begin{align}
	&0\leq c_{12}(y_1^n,y_2^n) \leq \min \{ c_1(y_1^n), c_2(y_2^n)\}  \ll 1\label{eq:smallclargen}
\end{align}
for every $(y_1^n,y_2^n)$.
It further follows from \eqref{eq:def_a12} that for every $(x^n,y_1^n,y_2^n)$,
\begin{align}
	\max\{a_1(x^n,y_1^n)+a_2(x^n,y_2^n)-1,0\} \leq  a_{12}(x^n,y_1^n,y_2^n) \leq \min\{a_1(x^n,y_1^n), a_2(x^n,y_2^n)\}.\label{eq:checkvalidprob}
\end{align}
Thus, \eqref{eq:validprob} is satisfied. To see this explicitly, for each $k\in \{1,2\}$, let $I_k(x^n,y_k^n)$ be the indicator of the predicate $(x^n,y_k^n) \in \mathcal{T}^{(n)}_{\epsilon}(\mP_X \mN_{Y_k\mid X})$, and let $\overline{I}_k(x^n,y_k^n)$ be the indicator of its negation. Then, for every $(x^n,y_1^n,y_2^n)$,
\begin{align}
	\begin{bmatrix}
		a_1 \\ a_2 \\ a_{12}
	\end{bmatrix}
	=
	\begin{bmatrix}
		1 \\ 1 \\ 1
	\end{bmatrix} I_1 I_2
	+
	\begin{bmatrix}
		c_1 \\ 1 \\ c_1
	\end{bmatrix} \overline{I}_1 I_2
	+
	\begin{bmatrix}
		1 \\ c_2 \\ c_2
	\end{bmatrix} I_1 \overline{I}_2
	+
	\begin{bmatrix}
		c_1 \\ c_2 \\ c_{12}
	\end{bmatrix} \overline{I}_1 \overline{I}_2.
\end{align}
In particular, note that since on the RHS, exactly one of $I_1I_2, \overline{I}_1 I_2, I_1\overline{I}_2, \overline{I}_1 \overline{I}_2$ is always equal to one, while the rest are zero, the vector $[a_1,a_2,a_{12}]^T$  always evaluates to exactly one of the four columns on the RHS. The first three realizations trivially satisfy \eqref{eq:checkvalidprob}, while the last column satisfies \eqref{eq:checkvalidprob} for sufficiently large $n$ due to \eqref{eq:smallclargen}.
Therefore, for every sufficiently small $\epsilon$, the above construction yields valid NS-assisted coding schemes $\mZ_{n}$ for all sufficiently large $n$.
 
\subsection{Reliability analysis}
We now analyze the probability of error of the scheme described above.
The transmitter provides the scheme with $(W_1,W_2)$. The scheme generates $X^n$ (a random seed) according to $\mP_X^{\otimes n}$. The seed is used as the input to the broadcast channel, and for each $k\in \{1,2\}$, User $k$ receives $Y_k^n$ (a key), which comes from the stochastic mapping $\mN_{Y_k\mid X}^{\otimes n}$ of the seed. The key is then used as the input to the scheme at User $k$. The law of large numbers then implies
\begin{align}
	\lim_{n\to \infty}\Pr\big( (X^n,Y_k^n) \in \mathcal{T}^{(n)}_{\epsilon}(\mP_X\mN_{Y_k\mid X}), \forall k\in \{1,2\} \big) = 1.
\end{align}
It follows that with probability $\to 1$,
\begin{align}
	a_{12}(X^n,Y_1^n,Y_2^n) = 1.
\end{align}
Therefore, with probability $\to 1$, the scheme provides the correct messages for both users. It follows that $\lim_{n\to \infty} P_{e,k}(\mZ_n)=0, \forall k\in \{1,2\}$.
The rate tuples $(R_1,R_2)$ thus achieved are exactly described by Eq. \eqref{eq:achi_rates}. Taking $\epsilon\to 0$ concludes the proof. 
\hfil \qed

\section{Conclusion}
A precise characterization of the capacity region of the $K$-user BC with NS-assistance may be surprising, considering that the classical capacity region is still among the most prominent long-standing open problems in network information theory. While it is natural to imagine that the new insights from the NS-assisted capacity region could be useful to make progress in the classical case, no direct connection is immediately apparent to us. The result does impart new operational meaning to Sato's region. For example, Tian and Shamai previously noted in \cite{Tian_Shamai_Sato} that the synergistic information in a partial information decomposition of \cite{Partial_Information_Decomposition} can be rigorously interpreted as a \emph{lower bound} on the cooperative gain of a corresponding broadcast channel, subject to a specified input distribution. The cooperative gain was interpreted through Sato's bound, which we now know to be a tight characterization of the NS-assisted capacity. Thus, the updated insight is that  the synergistic information in the partial information decomposition of \cite{Partial_Information_Decomposition} can be rigorously interpreted as \emph{exactly} the cooperative gain in the NS-assisted capacity of a corresponding broadcast channel, subject to a given input distribution.

\appendix
\section{Proof of Lemma \ref{lem:typical_off_K}} \label{proof:typical_off_K}
We are given $\epsilon\in (0,1), n\in \mathbb{N}, \mP_X\in \mathcal{P}(\mathcal{X})$ and for each $k\in [\kappa]$, $ \mP_{Y_k\mid X}\in \mathcal{P}(\mathcal{Y}_k\mid \mathcal{X}), y_k^n\in \mathcal{Y}_k^n$. Let $\mP_{XY_k} \triangleq \mP_X\mP_{Y_k\mid X}, \forall k\in [\kappa]$.
We recall the definitions of types and typical sets in \eqref{eq:def_type} and \eqref{eq:def_typical_set}. Let
\begin{align}
	&\widetilde{\mathcal{T}} \triangleq \big\{ x^n\colon (x^n,y_k^n) \in \mathcal{T}_{\epsilon}^{(n)}(\mP_{XY_k}), \forall k\in [\kappa]\}.
\end{align}
Write $p \triangleq \sum_{x^n \in \widetilde{\mathcal{T}}}\mP_{X}^{\otimes n}(x^n) = \Pr(X^n \in \widetilde{\mathcal{T}})$ where $X^n \sim \mP_{X}^{\otimes n}$. We shall prove that 
\begin{align*}
	p \leq 2^{-n(J(\mP_X, (\mP_{Y_k\mid X})_{k\in [\kappa]})-\delta(\epsilon))}.
\end{align*}
The statement is trivially true if $p=0$, so consider $p>0$.

We proceed with a change of measure argument.
Let $\widetilde{X}^n \sim \widetilde{\mP}_{\widetilde{X}^n}\in \mathcal{P}(\mathcal{X}^n)$ such that 
\begin{align}
	\widetilde{\mP}_{\widetilde{X}^n}(x^n) \triangleq 
	\begin{cases}
		\mP_{X}^{\otimes n}(x^n)/p, & x^n \in \widetilde{\mathcal{T}}\\
		0, & x^n \in \mathcal{X}^n \setminus  \widetilde{\mathcal{T}}
	\end{cases}.
\end{align}
In other words, the distribution 
\begin{align}
	\widetilde{\mP}_{\widetilde{X}^n}(x^n) = \Pr(X^n=x^n \mid X^n \in \widetilde{\mathcal{T}})
\end{align}
is equal to the conditional distribution of $X^n$ given that $X^n \in \widetilde{\mathcal{T}}$.
It follows that $p$ is related to a relative entropy,
\begin{align} \label{eq:change_of_measure_bound}
	D(\widetilde{\mP}_{\widetilde{X}^n} \Vert \mP_X^{\otimes n}) &= \sum_{x^n \in \widetilde{\mathcal{T}}} \widetilde{\mP}_{\widetilde{X}^n}(x^n) \log_2 \frac{\widetilde{\mP}_{\widetilde{X}^n}(x^n)}{\mP_{X}^{\otimes n}(x^n)} =\log_2\left(\frac{1}{p}\right).
\end{align}

Recall that in our notations,
\begin{align}
	\widetilde{X}^n = (\widetilde{X}_1,\ldots, \widetilde{X}_n), ~y_k^n = (y_{k,1}, \ldots, y_{k,n}).
\end{align}
Define random variables
\begin{align}
	\widetilde{X} \triangleq \widetilde{X}_T, ~ Y_k\triangleq y_{k,T}, \forall k\in [\kappa],
\end{align}
where $T$ is a time-sharing variable, uniformly distributed over $[n]$, and is independent of $\widetilde{X}^n$. The random variables $(\widetilde{X}^n,T,\widetilde{X},Y_1,\ldots, Y_{\kappa})$ are jointly distributed and their joint pmf is denoted by $\widetilde{\mP}_{\widetilde{X}^nT\widetilde{X}Y_1\cdots Y_{\kappa}}$.

Following \eqref{eq:change_of_measure_bound},
\begin{align}
	&\log_2\left(\frac{1}{p}\right)\notag \\
	&=D(\widetilde{\mP}_{\widetilde{X}^n} \Vert \mP_X^{\otimes n})  \\
	&\geq \sum_{i=1}^n D(\widetilde{\mP}_{\widetilde{X}_i}\Vert \mP_X)\label{eq:fullchainrule}\\
	&=n  \sum_{i=1}^n  \widetilde{\mP}_{T}(i) D(\widetilde{\mP}_{\widetilde{X}|T=i}\Vert \mP_X)\label{eq:use_uniforma}\\
	&= n D(\widetilde{\mP}_{\widetilde{X}|T}\Vert \mP_X~|~\widetilde{\mP}_{T})\\
	&= n D(\widetilde{\mP}_{\widetilde{X} T} \Vert \mP_{X}\widetilde{\mP}_T) \label{eq:use_chain_rule1}\\
	&= n D(\widetilde{\mP}_{T|\widetilde{X}} \Vert \widetilde{\mP}_T~|~ \widetilde\mP_{\widetilde{X}}) + nD(\widetilde{\mP}_{\widetilde{X}}\Vert \mP_X)  \label{eq:use_chain_ruleb}\\
	&= n I(\widetilde{X};T) + nD(\widetilde{\mP}_{\widetilde{X}}\Vert \mP_X) \label{eq:use_mutualinfo} \\
	&\geq n I(\widetilde{X};T)  \label{eq:use_chain_rule2a} \\
	&\geq nI(\widetilde{X};Y_1,\ldots, Y_{\kappa}) \label{eq:use_ng}
\end{align}
where Step \eqref{eq:fullchainrule} follows from \cite[Theorem 2.16(c)]{Polyanskiy_book},  Step \eqref{eq:use_uniforma} is because $\widetilde{\mP}_T(i) = 1/n, \forall i\in [n]$. Step \eqref{eq:use_chain_rule1} follows from \cite[Theorem 2.16(a)]{Polyanskiy_book}. Step \eqref{eq:use_chain_ruleb} follows from the chain rule of relative entropy 
\cite[Theorem 2.15]{Polyanskiy_book}. Step \eqref{eq:use_mutualinfo} follows from \cite[Theorem 3.2(a)]{Polyanskiy_book}. 
Step \eqref{eq:use_chain_rule2a} follows from the non-negativity of relative entropy. Step \eqref{eq:use_ng} follows from the data processing inequality and the fact that $(Y_1,\ldots, Y_{\kappa})$ is a deterministic function of $T$.

Next we relate $I(\widetilde{X};Y_1,\ldots, Y_{\kappa})$ to $J(\mP_X, (\mP_{Y_k\mid X})_{k\in[\kappa]})$. Since $\epsilon$ is given by the lemma, we use $\epsilon'$ in the following definitions of a set and a function with respect to $\epsilon'$. For $\epsilon'\geq 0$, define the set
\begin{equation}
\begin{aligned}\label{eq:defQepsilon}
	&\mathcal{Q}_{\epsilon'} \triangleq \Big\{ \mQ_{XY_1\cdots Y_{\kappa}} \in \mathcal{P}(\mathcal{X}\times \mathcal{Y}_1\times \cdots \times \mathcal{Y}_{\kappa}) \colon \\
	&|\mQ_{XY_k}(x,y_k) - \mP_{XY_k}(x,y_k)| \leq \epsilon' \mP_{XY_k}(x,y_k), \forall k\in [\kappa] \Big\}.
\end{aligned}
\end{equation}
Accordingly, define the function
\begin{align}
	J_{\epsilon'} \triangleq \min_{\mQ_{XY_1\cdots Y_{\kappa}} \in \mathcal{Q}_{\epsilon'}} I_{\mQ_{XY_1\cdots Y_{\kappa}}}(X;Y_1,\ldots, Y_{\kappa}).\label{eq:defJepsilon}
\end{align}
Observe that for every $0\leq \epsilon_1' < \epsilon'_2$, $\mathcal{Q}_{\epsilon_1'} \subseteq \mathcal{Q}_{\epsilon_2'}$, so we have that $J_{\epsilon'_1} \geq J_{\epsilon'_2}$, i.e., $J_{\epsilon'}$ is monotonically decreasing. Moreover, we have the following lemma.

\begin{lemma}\label{lem:Jepsilon}
\begin{align}
\lim_{\epsilon' \rightarrow 0+}J_{\epsilon'}=J_0.
\end{align}
\end{lemma}
\begin{proof} Since $J_{\epsilon'}\leq J_0$ for all ${\epsilon'}\geq 0$, it suffices to prove the claim $\liminf_{{\epsilon'}\rightarrow 0+}J_{\epsilon'}=J_0$. To set up a proof by contradiction, suppose this claim is not true, i.e., $\liminf_{{\epsilon'}\rightarrow 0+}J_{\epsilon'}=J_0-\eta$ for some $\eta>0$. Thus, there exists a sequence $\epsilon'_m>0$, $m\in\mathbb{N}$ that approaches $0$ as $m\rightarrow\infty$ such that $J_{\epsilon'_m}\leq J_0-\eta$ for all $m\in\mathbb{N}$. Noting that $J_{\epsilon'_m}$ is defined as the minimum over the compact set $\mathcal{Q}_{\epsilon'_m}$, let $\mQ^{(m)}\in\mathcal{Q}_{\epsilon'_m}$ be a corresponding minimizer, for all $m\in\mathbb{N}$. Recall that the probability simplex over a finite alphabet is compact, and every sequence in a compact set has a convergent subsequence whose limit lies in the same set. Thus, there is a subsequence of $\mQ^{(m)}$ (which with a small abuse of notation we denote also as $\mQ^{(m)}$), such that the $\mQ^{(m)}\rightarrow \mQ^*$. Moreover since for all $m\in\mathbb{N}$, we have $\mQ^{(m)}\in\mathcal{Q}_{\epsilon'_m}$, i.e., $|\mQ^{(m)}(x,y_k) - \mP_{XY_k}(x,y_k)| \leq \epsilon'_m \mP_{XY_k}(x,y_k)$, and $\epsilon'_m\rightarrow 0$, it follows that the limiting value $Q^*\in\mathcal{Q}_0= \Big\{ \mQ_{XY_1\cdots Y_{\kappa}}   \colon  \mQ_{XY_k} = \mP_{XY_k},\forall k\in [\kappa] \Big\}$. Thus, $J_0-\eta\geq I_{\mQ^*}(X;Y_1,\dots,Y_\kappa)\geq \min_{\mQ\in\mathcal{Q}_0}I_{\mQ}(X;Y_1,\dots,Y_\kappa)=J_0$. The contradiction completes the proof.
\end{proof}
It follows that there exists $\delta(\epsilon') \geq 0$ which tends to $0$ as $\epsilon' \to 0$ such that 
\begin{align}
	J_{\epsilon'} &= J_0 - \delta(\epsilon') \\
	&= J(\mP_X,(\mP_{Y_k\mid X})_{k\in [\kappa]})- \delta(\epsilon') \label{eq:Jepsilon_with_J0}
\end{align}
where \eqref{eq:Jepsilon_with_J0} is due to \eqref{eq:defJepsilon} and \eqref{eq:def_JK}.

Next we show that $\widetilde{\mP}_{\widetilde{X} Y_1 \cdots Y_K} \in \mathcal{Q}_{\epsilon}$.
 For each $k\in [\kappa]$, we compute the marginal distribution for $(\widetilde{X}, Y_k)$.
For every $(x,y_k)\in \mathcal{X}\times \mathcal{Y}_k$,
\begin{align}
	&\widetilde{\mP}_{\widetilde{X} Y_k}(x,y_k) \notag \\
	&=\frac{1}{n} \sum_{i=1}^n \widetilde{\mP}_{\widetilde{X}Y_k \mid T=i}(x,y_k) \\
	&= \frac{1}{n}\sum_{i=1}^n\widetilde{\mP}_{\widetilde{X}_i}(x) \mathbb{I}(y_{k,i}=y_k) \\
	&= \frac{1}{n} \sum_{i=1}^n \sum_{x^n \in \widetilde{\mathcal{T}}} \widetilde{\mP}_{\widetilde{X}^n}(x^n) \mathbb{I}(x_i=x,y_{k,i}=y_k) \\
	&= \sum_{x^n \in \widetilde{\mathcal{T}}} \widetilde{\mP}_{\widetilde{X}^n}(x^n) \Big(\frac{1}{n} \sum_{i=1}^n\mathbb{I}(x_i=x,y_{k,i}=y_k) \Big)\\
	&= \sum_{x^n \in \widetilde{\mathcal{T}}} \widetilde{\mP}_{\widetilde{X}^n}(x^n) \widehat{\mP}_{x^n,y_k^n}(x,y_k)
\end{align}
By definition, for each $x^n \in \widetilde{\mathcal{T}}$, we have $(x^n,y_k^n) \in \mathcal{T}_{\epsilon}^{(n)}(\mP_{XY_k})$, which means
\begin{align}
	|\widehat{\mP}_{x^n,y_k^n}(x,y_k) - \mP_{XY_k}(x,y_k)| \leq \epsilon \mP_{XY_k}(x,y_k), ~~\forall (x,y_k) \in \mathcal{X}\times \mathcal{Y}_k.
\end{align}
It follows that, $\forall (x,y_k) \in \mathcal{X}\times \mathcal{Y}_k$, 
\begin{align}
	&|\widetilde{\mP}_{\widetilde{X} Y_k}(x,y_k) - \mP_{XY_k}(x,y_k)| \notag \\
	&\leq \sum_{x^n \in \widetilde{\mathcal{T}}}\widetilde{\mP}_{\widetilde{X}^n}(x^n)|\widehat{\mP}_{x^n,y_k^n}(x,y_k) - \mP_{XY_k}(x,y_k)| \label{eq:use_jensen} \\
	&\leq \epsilon \mP_{XY_k}(x,y_k)
\end{align}
where \eqref{eq:use_jensen} applies Jensen's inequality to the function $f(z)=|z|$.
Therefore, by definition $\widetilde{\mP}_{\widetilde{X}Y_1\cdots Y_{\kappa}} \in \mathcal{Q}_{\epsilon}$. Since $J_{\epsilon}$ is defined as the minimum of mutual information over $\mathcal{Q}_{\epsilon}$, we have,
\begin{align}
	I(\widetilde{X};Y_1,\ldots, Y_{\kappa}) \geq  J_{\epsilon}  \stackrel{\eqref{eq:Jepsilon_with_J0}}{=} J(\mP_X,(\mP_{Y_k\mid X})_{k\in [\kappa]}) - \delta(\epsilon).\label{eq:Ibound}
\end{align}

It follows from \eqref{eq:use_ng} and \eqref{eq:Ibound} that
\begin{align}
	p \leq 2^{-nI(\widetilde{X};Y_1,\ldots, Y_{\kappa})} \leq 2^{-nJ_{\epsilon}} = 2^{-n (J(\mP_X,(\mP_{Y_k\mid X})_{k\in [\kappa]})-\delta(\epsilon))}.
\end{align}
We thus conclude the proof. \hfil \qed

\section{Proof of Theorem \ref{thm:NSBC_Kuser}: Achievability}\label{proof:NSBC_Kuser}
We prove the achievability by generalizing the authentication solution used in Section \ref{proof:NSBC_capacity}.
\subsection{Towards a $K$-user authentication solution}
In the following, we write $J(\mP_X,(\mN_{Y_k\mid X})_{k\in \mathcal{K}})$ simply as $J_{\mathcal{K}}$, for each $\emptyset \neq \mathcal{K}\subseteq [K]$.
We shall prove that, for each $\mP_X\in \mathcal{P}(\mathcal{X})$, $(R_1,\ldots,R_K)\in \mathbb{R}_{> 0}^K$ is achievable if
\begin{align} \label{eq:rate_K}
	\sum_{k\in \mathcal{K}} R_k < J_{\mathcal{K}}, \quad\forall \emptyset \neq \mathcal{K}\subseteq [K]
\end{align}
whenever $\min_{k\in [K]} \{I  (X;Y_k)\} > 0$. Note that if, for example, $I(X;Y_1)=0$, then the achievability of any $(0,R_2,\ldots, R_K)$ satisfying \eqref{eq:rate_K} is implied by the degenerate $(K-1)$-user solution where User $1$ is absent.

Firstly, the authentication solution, when extended to the $K$-user case, has the form
\begin{align}
	&\mZ(x^n, \widehat{w}_1,\ldots, \widehat{w}_K\mid (w_1,\ldots, w_K), y_1^n,\ldots, y_K^n)\notag\\
	&= \mZ(x^n) \mZ(\widehat{w}_1,\ldots, \widehat{w}_K\mid (w_1,\ldots, w_K), x^n, y_1^n,\ldots, y_K^n)\label{eq:impliesKC0}
\end{align}
where $\mZ(x^n)$ specifies a pmf for $X^n$. To specify $\mZ(\widehat{w}_1,\ldots, \widehat{w}_K\mid (w_1,\ldots, w_K), x^n, y_1^n,\ldots, y_K^n)$, for each $(x^n,y_1^n,\ldots, y_K^n)$,  let $(A_1,A_2,\ldots, A_K)$ be a tuple of jointly distributed binary random variables drawn according to a distribution $\mP^{(x^n,y_1^n,\ldots, y_K^n)}_{A_1\cdots A_K}\in \mathcal{P}(\{0,1\}^K)$. Specifically, $A_k=1$ means the scheme produces the correct message for User $k$, and $A_k=0$ means the scheme produces an incorrect message, chosen randomly at uniform, for User $k$.  Then for $\mathcal{K}\subseteq [K]$, and when $(\widehat{w}_k = w_k, \forall k\in \mathcal{K} ~\mbox{and} ~ \widehat{w}_{k'} \neq w_{k'}, \forall k'\in [K]\setminus \mathcal{K})$, we have
\begin{align}
	&\mZ(\widehat{w}_1,\ldots, \widehat{w}_K\mid (w_1,\ldots, w_K), x^n, y_1^n,\ldots, y_K^n) \notag \\
	&= \frac{\mP^{(x^n,y_1^n,\ldots, y_K^n)}_{A_1\cdots A_K}\big((A_k=1)_{k\in \mathcal{K}},(A_{k'}=0)_{k'\in [K]\setminus \mathcal{K}}\big)}{\prod_{k'\in [K]\setminus \mathcal{K}}(M_{k'}-1)} 
\end{align}
Besides the validity of $\mP^{(x^n,y_1^n,\ldots, y_K^n)}_{A_1\cdots A_K}$ as a probability distribution pointwise for each $(x^n,y_1^n,\ldots, y_K^n)$, we further require the following conditions to hold. 
\begin{enumerate}[label=P\arabic*:]
	\item For each $k \in [K]$, the marginal distribution $\mP^{(x^n,y_1^n,\ldots, y_K^n)}_{A_1\cdots A_{k-1}A_{k+1}\cdots A_K}$ (without $A_k$) does not depend on $y_k^n$;
	\item For each $\mathcal{K} \subseteq [K]$, let 
	\begin{align}
	a_{\mathcal{K}}(x^n,y_1^n,\ldots, y_K^n) \triangleq \Pr((A_k = 1)_{\forall k\in \mathcal{K}} \mid (x^n,y_1^n,\ldots, y_K^n)),
	\end{align}
	with the convention that $a_{\emptyset}(x^n,y_1^n,\ldots, y_K^n) \triangleq 0$. Then, for every non-empty set $\mathcal{K} \subseteq [K]$,
\begin{align}
	&\mathbb{E}_{X^n \sim \mZ(x^n)}\big[a_{\mathcal{K}}(X^n,y_1^n,\ldots, y_K^n)\big] \notag \\
	&=	\mathbb{E}_{X^nA_1\cdots A_K\mid (y_1^n,\ldots, y_K^n)}\Big[ \prod_{k\in \mathcal{K}}A_k \Big] \\
	&= \prod_{k\in \mathcal{K}}\frac{1}{M_k},
\end{align}
for all $(y_1^n,\ldots, y_K^n)$. 
\end{enumerate}
It is noteworthy that conditioned on $(x^n,y_1^n,\ldots, y_K^n)$, the values $a_{\mathcal{K}}$ correspond to the definition of the `\emph{M\"obius parameter associated with $\mathcal{K}$},' albeit with $0$ and $1$ switched, in \cite[Equation (10)]{Mobius}.
Observe that $a_{\mathcal{K}}(x^n,y_1^n,\ldots, y_K^n)$ is determined by the marginal distribution $\mP_{(A_k)_{k\in \mathcal{K}}}^{(x^n,y_1^n,\ldots, y_k^n)}$, i.e., the marginal for $(A_k)_{k\in \mathcal{K}}$. Therefore, P1 implies that $a_{\mathcal{K}}(x^n,y_1^n,\ldots, y_K^n)= a_{\mathcal{K}}(x^n,(y_k^n)_{k\in \mathcal{K}})$, which depends only on $(x^n,(y_k^n)_{k\in \mathcal{K}})$.

We claim that any $\mZ(x^n, \widehat{w}_1,\ldots, \widehat{w}_K\mid (w_1,\ldots, w_K), y_1^n,\ldots, y_K^n)$ thus specified satisfies the non-signaling conditions \textit{C0} through \textit{CK}.  Let us first verify \textit{C1}.
\begin{align}
	&\sum_{\widehat{w}_1}\mZ(x^n,\widehat{w}_1,\ldots, \widehat{w}_K\mid (w_1,\ldots, w_K),  y_1^n,\ldots, y_K^n)\notag \\
	&= \mZ(x^n)\sum_{\widehat{w}_1}\mZ(\widehat{w}_1,\ldots, \widehat{w}_K\mid (w_1,\ldots, w_K), x^n, y_1^n,\ldots, y_K^n)\label{eq:checkingC1}
\end{align} 
Suppose, without loss of generality, $\widehat{w}_k\neq w_k, \forall 2\leq k\leq \ell$ and $\widehat{w}_{k'}= w_{k'}, \forall k' > \ell$. Then we have
\begin{align}
	\eqref{eq:checkingC1} &=\mZ(x^n)\frac{\mP^{(x^n,y_1^n,\ldots, y_K^n)}_{A_1\cdots A_K}\big(1,\overbrace{0,\ldots, 0}^{\ell-1}, 1,\ldots, 1)}{\prod_{k'=2}^{\ell} (M_{k'}-1)} \notag \\
	&\quad+\mZ(x^n)(M_1-1) \frac{\mP^{(x^n,y_1^n,\ldots, y_K^n)}_{A_1\cdots A_K}\big(0,\overbrace{0,\ldots, 0}^{\ell-1}, 1,\ldots, 1)}{\prod_{k'=1}^{\ell} (M_{k'}-1)} \\
	&=\mZ(x^n)\frac{\mP^{(x^n,y_1^n,\ldots, y_K^n)}_{A_2\cdots A_K}\big(\overbrace{0,\ldots, 0}^{\ell-1}, 1,\ldots, 1)}{\prod_{k'=2}^{\ell} (M_{k'}-1)}
\end{align}
which does not depend on $y_1^n$ by the first condition. By symmetry, this argument also shows that \textit{Ck} holds for $k\in \{2,\ldots, K\}$.

We next verify \textit{C0}. For $\mathcal{K}\subseteq [K]$, and when $(\widehat{w}_k = w_k, \forall k\in \mathcal{K} ~\mbox{and} ~ \widehat{w}_{k'} \neq w_{k'}, \forall k'\in [K]\setminus \mathcal{K})$, write $q = 1/\prod_{k'\in [K]\setminus \mathcal{K}}(M_{k'}-1)$, and we have
\begin{align}
	&\sum_{x^n}\mZ(x^n,\widehat{w}_1,\ldots, \widehat{w}_K\mid (w_1,\ldots, w_K),  y_1^n,\ldots, y_K^n)\notag \\
	&=q\cdot \mathbb{E}_{X^n \sim \mZ(x^n)} \Big[\mP^{(X^n,y_1^n,\ldots, y_K^n)}_{A_1\cdots A_K}\big((A_k=1)_{k\in \mathcal{K}},(A_{k'}=0)_{k'\in [K]\setminus \mathcal{K}}\big)\Big]\\
	&=q\cdot\mathbb{E}_{X^n \sim \mZ(x^n)} \mathbb{E}_{A_1\cdots A_K\mid (x^n,y_1^n,\ldots, y_K^n)}  \Big[\prod_{k\in \mathcal{K}} A_k \prod_{k'\in [K]\setminus \mathcal{K}} (1-A_{k'}) \Big]\\
	&=q\cdot \mathbb{E}_{X^nA_1\cdots A_K\mid (y_1^n,\ldots, y_K^n)}\Big[\prod_{k\in \mathcal{K}} A_k \prod_{k'\in [K]\setminus \mathcal{K}} (1-A_{k'}) \Big] \label{eq:expectation}
\end{align}
Due to the linearity of expectation, the expectation in \eqref{eq:expectation} can be expanded as some linear combination of $\mathbb{E}[\prod_{k\in \mathcal{S}}A_k]$ for $\emptyset\neq \mathcal{S}\subseteq [K]$, i.e., there exist some constants $\alpha_{\mathcal{S}}$ such that,
\begin{align}
q\cdot \mathbb{E}_{X^nA_1\cdots A_K\mid (y_1^n,\ldots, y_K^n)}\Big[\prod_{k\in \mathcal{K}} A_k \prod_{k'\in [K]\setminus \mathcal{K}} (1-A_{k'}) \Big] &=\sum_{\emptyset\neq \mathcal{S}\subseteq [K]}\alpha_{\mathcal{S}}\mathbb{E}[\prod_{k\in \mathcal{S}}A_k].
\end{align}
Let $(\widetilde{A}_1,\ldots, \widetilde{A}_K)$ be $K$ independent random variables such that $\mathbb{E}[\widetilde{A}_k]=1/M_k$ for all $k\in[K]$. Observe that  
\begin{align}
\mathbb{E}[\prod_{k\in \mathcal{S}}\widetilde{A}_k] = \frac{1}{\prod_{k\in \mathcal{S}}M_k}= \mathbb{E}[\prod_{k\in \mathcal{S}}A_k] , \quad\forall \emptyset \neq \mathcal{S} \subseteq [K].
\end{align}
Then we have,
\begin{align}
	&\sum_{\emptyset\neq \mathcal{S}\subseteq [K]}\alpha_{\mathcal{S}}\mathbb{E}[\prod_{k\in \mathcal{S}}A_k]\notag\\
	&=\sum_{\emptyset\neq \mathcal{S}\subseteq [K]}\alpha_{\mathcal{S}}\mathbb{E}[\prod_{k\in \mathcal{S}}\widetilde{A}_k]\\
	 &=q\cdot \mathbb{E}\Big[\prod_{k\in \mathcal{K}} \widetilde{A}_k \prod_{k'\in [K]\setminus \mathcal{K}} (1-\widetilde{A}_{k'}) \Big]\\
	&= q\cdot \prod_{k\in \mathcal{K}}\mathbb{E}[\widetilde{A}_k] \prod_{k'\in [K]\setminus \mathcal{K}} (1-\mathbb{E}[\widetilde{A}_{k'}])\\
	&= q \cdot \prod_{k\in \mathcal{K}}\frac{1}{M_k}\prod_{k'\in [K]\setminus \mathcal{K}} \frac{M_{k'}-1}{M_{k'}} \\
	&= \prod_{k\in [K]}\frac{1}{M_k}
\end{align}
which does not depend on $(w_1,\ldots, w_K)$. Thus, $C0$ is satisfied.

\subsection{$K$-user authentication solution}
Recall that we have fixed $\mP_X$. Now let $\mZ(x^n) = \mP_{X}^{\otimes n}$. 
It remains to specify $\mP^{(x^n,y_1^n,\ldots, y_K^n)}_{A_1\cdots A_K}$. To this end, we will first specify the values 
\begin{align*}
	a_{\mathcal{K}}(x^n,(y_k^n)_{k\in \mathcal{K}}), \quad\forall \emptyset \neq \mathcal{K} \subseteq [K]
\end{align*}
pointwise for each $(x^n,y_1^n,\ldots, y_K^n)$,
and then show that the values yield valid $\mP^{(x^n,y_1^n,\ldots, y_K^n)}_{A_1\cdots A_K}$ (for sufficiently large $n$). 
Let $\overline{p}$ denote the negation of the predicate $p$.  Let $p\wedge q$ denote the logical conjunction of the predicates $p$ and $q$.

Let $\epsilon \in (0,1)$. For each $k\in [K]$, let $I_k(x^n, y_k^n)$ denote the predicate
\begin{align*}
	(x^n, y_k^n) \in \mathcal{T}^{(n)}_{\epsilon} (\mP_X\mN_{Y_k\mid X}).
\end{align*}
Define the constant $t_{\emptyset} = 1$, and for $\emptyset \neq \mathcal{K} \subseteq [K]$, let
\begin{align}
	t_{\mathcal{K}}((y_k^n)_{k\in \mathcal{K}}) \triangleq \Pr\Big(  \bigwedge_{k\in \mathcal{K}} I_k(X^n,y_k^n) \Big),
\end{align}
where $X^n \sim \mP_{X}^{\otimes n}$.

Define the constant $c_{\emptyset} \triangleq 1$. For $\emptyset \neq \mathcal{K} \subseteq [K]$, let $c_{\mathcal{K}}((y_k^n)_{k\in \mathcal{K}})$ be a function that depends only on $(y_k^n)_{k\in \mathcal{K}}$. 
We will iteratively define the values of $c_{\mathcal{K}}((y_k^n)_{k\in \mathcal{K}})$ and $a_{\mathcal{K}}(x^n,(y_k^n)_{k\in \mathcal{K}})$ for $K$ levels (steps). 
Specifically, for the level $L=1,2,\cdots, K$, and for each $\mathcal{K}\subseteq [K]$ such that $|\mathcal{K}| = L$, define, for each case indicated by a subset $\mathcal{S} \subseteq \mathcal{K}$ (thus $2^L$ cases),
\begin{align}
	a_{\mathcal{K}}(x^n, (y_k^n)_{k\in \mathcal{K}}) \triangleq c_{\mathcal{S}}((y_s^n)_{s\in \mathcal{S}}), ~ \mbox{if}~  \bigwedge_{s \in \mathcal{S}}\overline{I}_s(x^n,y_s^n) \bigwedge_{s'\in \mathcal{K} \setminus \mathcal{S}} I_{s'}(x^n,y_{s'}^n) 
\end{align}
where $c_{\mathcal{K}}((y_s^n)_{s\in \mathcal{K}})$ is chosen to meet the condition $\mathbb{E}_{X^n \sim \mP_{X}^{\otimes n}}\big[a_{\mathcal{K}}(X^n,y_1^n,\ldots, y_K^n)\big] = \frac{1}{\prod_{k\in \mathcal{K}}M_k}$. Note that unless $\mathcal{S} =\mathcal{K}$, $c_{\mathcal{S}}((y_s^n)_{s\in \mathcal{S}})$ has already been chosen before the $L^{th}$ level.

Let us work through a concrete example. Say $K=3$. We will  suppress the parameters $(x^n,y_1^n,\ldots, y_K^n)$ for notational conciseness.
\begin{align}
	&L=1:~~ \forall k\in \{1,2,3\}, \notag \\
	& a_{\{k\}} \triangleq
	\begin{cases}
		1, & \mbox{if}~ I_k\\
		c_{\{k\}}, & \mbox{if}~ \overline{I}_k
	\end{cases}\\
	&c_{\{k\}}\triangleq \frac{\frac{1}{M_k} - t_{\{k\}}}{1-t_{\{k\}}}
\end{align}
\begin{align}
	&L = 2\colon ~~\forall (i,j) \in \{(1,2),(1,3),(2,3)\}, \notag\\
	& a_{\{i,j\}} \triangleq \begin{cases}
		1, & \mbox{if}~ I_{i} \wedge I_{j} \\
		c_{\{i\}}, & \mbox{if}~ \overline{I}_{i} \wedge I_{j} \\
		c_{\{j\}}, & \mbox{if}~ I_{i} \wedge \overline{I}_{j} \\
		c_{\{i,j\}}, & \mbox{if}~ \overline{I}_{i} \wedge \overline{I}_{j}
	\end{cases} \\
	& c_{\{i,j\}} \triangleq \frac{\frac{1}{M_{i}M_{j}}-t_{\{i,j\}}-(t_{\{j\}}-t_{\{i,j\}})c_{\{i\}} - (t_{\{i\}}-t_{\{i,j\}})c_{\{j\}}}{1-t_{\{i\}}-t_{\{j\}}+t_{\{i,j\}}}
\end{align}
\begin{align}
	&L = 3, \notag \\
	&a_{\{1,2,3\}} \triangleq \begin{cases}
		1, & \mbox{if}~ I_1 \wedge I_2 \wedge I_3 \\
		c_{\{1\}}, & \mbox{if}~ \overline{I}_1 \wedge I_2 \wedge I_3 \\
		c_{\{2\}}, & \mbox{if}~ I_1 \wedge \overline{I}_2 \wedge I_3 \\
		c_{\{3\}}, & \mbox{if}~ I_1 \wedge I_2 \wedge \overline{I}_3 \\
		c_{\{1,2\}}, & \mbox{if}~ \overline{I}_1 \wedge \overline{I}_2 \wedge I_3 \\
		c_{\{1,3\}}, & \mbox{if}~ \overline{I}_1 \wedge I_2 \wedge \overline{I}_3 \\
		c_{\{2,3\}}, & \mbox{if}~ I_1 \wedge \overline{I}_2 \wedge \overline{I}_3 \\
		c_{\{1,2,3\}}, & \mbox{if}~ \overline{I}_1 \wedge \overline{I}_2 \wedge \overline{I}_3
	\end{cases}\\
	&c_{\{1,2,3\}} \triangleq \frac{\frac{1}{M_1M_2M_3} - \sum_{\mathcal{S}\subset \{1,2,3\}}c_{\mathcal{S}}\sum_{\mathcal{A} \subseteq \mathcal{S}}(-1)^{|\mathcal{A}|}t_{\mathcal{A} \cup (\{1,2,3\}\setminus \mathcal{S})}}{1-t_{\{1\}}-t_{\{2\}} - t_{\{3\}} + t_{\{1,2\}}+ t_{\{1,3\}}+t_{\{2,3\}}-t_{\{1,2,3\}}}
\end{align}
 Since $t_{\mathcal{K}}$ only depends on $(y_k^n)_{k\in \mathcal{K}}$, it can be readily verified that in general $c_{\mathcal{K}}$ only depends on $(y_k^n)_{k\in \mathcal{K}}$ and accordingly $a_{\mathcal{K}}$ only depends on $(x^n,(y_k^n)_{k\in \mathcal{K}})$. Generally, the condition $\mathbb{E}_{X^n \sim \mP_{X}^{\otimes n}}\big[a_{\mathcal{K}}(X^n,y_1^n,\ldots, y_K^n)\big] = \frac{1}{\prod_{k\in \mathcal{K}}M_k}$ leads to the recurrence formula 
\begin{align} \label{eq:def_c}
	c_{\mathcal{K}} = \frac{\frac{1}{\prod_{k\in \mathcal{K}}M_k}- \sum_{\mathcal{S} \subset \mathcal{K}}c_{\mathcal{S}} \sum_{\mathcal{A} \subseteq \mathcal{S}} (-1)^{|\mathcal{A}|}t_{\mathcal{A} \cup (\mathcal{K} \setminus \mathcal{S})}}{\sum_{\mathcal{A} \subseteq \mathcal{K}}(-1)^{|\mathcal{A}|}t_{\mathcal{A}}}.
\end{align}
The following steps are used to derive \eqref{eq:def_c}.
\begin{align}
&\mathbb{E}_{X^n \sim \mP_{X}^{\otimes n}}\big[a_{\mathcal{K}}(X^n,y_1^n,\ldots, y_K^n)\big] \\
&= \sum_{\mathcal{S} \subseteq \mathcal{K}} c_{\mathcal{S}}  \Pr\Big(\bigwedge_{s \in \mathcal{S}}\overline{I}_s(X^n, y_s^n)  \bigwedge_{s'\in \mathcal{K} \setminus \mathcal{S}} I_{s'}(X^n, y_{s'}^n)  \Big)\\
&= \sum_{\mathcal{S} \subseteq \mathcal{K}} c_{\mathcal{S}} \mathbb{E} \Big[\prod_{s\in \mathcal{S}}(1-I_s) \prod_{s'\in \mathcal{K} \setminus \mathcal{S}} I_{s'} \Big] \label{eq:introduce_rvI} \\
&= \sum_{\mathcal{S} \subseteq \mathcal{K}} c_{\mathcal{S}}  \sum_{\mathcal{A} \subseteq \mathcal{S}} (-1)^{|\mathcal{A}|} \mathbb{E} \Big[ \prod_{s\in \mathcal{A}\cup (\mathcal{K} \setminus \mathcal{S})} I_s \Big]\\
&= \sum_{\mathcal{S} \subseteq \mathcal{K}} c_{\mathcal{S}}  \sum_{\mathcal{A} \subseteq \mathcal{S}} (-1)^{|\mathcal{A}|} \Pr \Big( \bigwedge_{s\in \mathcal{A} \cup (\mathcal{K} \setminus \mathcal{S})} I_s(X^n,y_s^n) \Big) \\
& = \sum_{\mathcal{S} \subseteq \mathcal{K}} c_{\mathcal{S}}  \sum_{\mathcal{A} \subseteq \mathcal{S}} (-1)^{|\mathcal{A}|} t_{\mathcal{A} \cup (\mathcal{K} \setminus \mathcal{S})} \label{eq:to_equal_M_inverse}
\end{align}
In Step \eqref{eq:introduce_rvI}, with a slight abuse of notation, we introduce $I_s$, for each $s\in \mathcal{K}$, as the binary random variable which equals $1$ if $I_s(X^n,y_{s}^n)$ is true, and $0$ otherwise. Setting \eqref{eq:to_equal_M_inverse} to $\frac{1}{\prod_{k\in \mathcal{K}}M_k}$ and rearranging the terms yields the formula in \eqref{eq:def_c}.

\subsection{Validity of the solution}
Applying Lemma \ref{lem:typical_off_K} with respect to each non-empty subset of users, we have that, for each non-empty subset $\mathcal{K}\subseteq [K]$, there exists $\delta_{\mathcal{K}}(\epsilon) \to 0$ as $\epsilon \to 0$ such that
\begin{align}
	t_{\mathcal{K}}((y_k^n)_{k\in \mathcal{K}}) \leq 2^{-n(J_{\mathcal{K}}-\delta_{\mathcal{K}}(\epsilon))}
\end{align}
for every $(y_k^n)_{k\in \mathcal{K}}$.
From now on, consider $\epsilon$ to be sufficiently small so that 
\begin{align}
	J_{\mathcal{K}} -\delta_{\mathcal{K}}(\epsilon) > 0, ~~\forall \emptyset \neq \mathcal{K} \subseteq [K].
\end{align}
For any $(R_1,\ldots, R_K)\in \mathbb{R}_{>0}^K$ such that
\begin{align} \label{eq:achi_rates_K}
	\sum_{k\in \mathcal{K}} R_k < J_{\mathcal{K}} - \delta_{\mathcal{K}}(\epsilon),~~ \forall \emptyset \neq \mathcal{K} \subseteq [K]
\end{align}
by picking
\begin{align}
	M_k \triangleq \lfloor 2^{nR_k} \rfloor, ~~\forall k\in [K],
\end{align}
it is guaranteed that, as $n\to \infty$,
\begin{align}
	\frac{\log_2(M_k)}{n} \to R_k,~~ \forall k\in [K].
\end{align}
Moreover, since
\begin{align}
	t_\mathcal{K} \leq 2^{-n(J_{\mathcal{K}}-\delta_{\mathcal{K}}(\epsilon))} < 2^{-n\sum_{k\in \mathcal{K}}R_k} \leq \frac{1}{\prod_{k\in \mathcal{K}}M_k}
\end{align}
we have
\begin{align}
	t_{\mathcal{K}} = o\left(\frac{1}{\prod_{k\in \mathcal{K}}M_k}\right),\quad \forall \emptyset \neq \mathcal{K} \subseteq [K].
\end{align}
It can now be proved by induction on $L=1,\ldots, K$ that
\begin{align}
	c_{\mathcal{K}} \simeq \frac{1}{\prod_{k\in \mathcal{K}}M_k}, ~~\forall \emptyset \neq \mathcal{K} \subseteq [K]. \label{eq:c_approx}
\end{align}
Here, $g(n) = o(f(n))$ means that $\frac{g(n)}{f(n)} \to 0$, and $g(n) \simeq f(n)$ means that $\frac{g(n)}{f(n)} \to 1$, as $n\to \infty$.
To see this, we first observe that for each singleton $\mathcal{K}=\{k\}$, \eqref{eq:def_c} implies $c_{\{k\}} = \frac{1/M_k-t_{\{k\}}}{1-t_{\{k\}}} \simeq 1/M_k$. Assume that this holds for every $\mathcal{K}$ such that $|\mathcal{K}|\leq L$. Then, for any $\mathcal{K}$ such that $|\mathcal{K}|=L+1$, it suffices to show that the denominator of \eqref{eq:def_c} $\to 1$ and the numerator of \eqref{eq:def_c} $\simeq \frac{1}{\prod_{k\in \mathcal{K}}M_k}$. Note that the denominator is a finite sum, whose $\mathcal{A}=\emptyset$ term equals 1. All remaining terms are absolutely upper bounded by exponentially decaying functions and therefore vanish asymptotically. Hence, the denominator converges to 1. In the numerator, the second term is a finite sum, in which for each $\mathcal{S}\subset \mathcal{K}$,
\begin{align*}
	c_{\mathcal{S}} \simeq \frac{1}{\prod_{k\in \mathcal{S}}M_k}
\end{align*}
due to the inductive hypothesis, and each term
\begin{align*}
	(-1)^{|\mathcal{A}|} t_{\mathcal{A} \cup (\mathcal{K}\setminus \mathcal{S})}
\end{align*}
 of the inner sum is absolutely upper bounded by $o(\frac{1}{\prod_{k\in \mathcal{K}\setminus \mathcal{S}}M_k})$. Therefore, the second term in the numerator is absolutely upper bounded by $o(\frac{1}{\prod_{k\in \mathcal{K}}M_k})$. It follows that the numerator $\simeq \frac{1}{\prod_{k\in \mathcal{K}}M_k}$.

It remains to show that, for sufficiently large $n$, the resulting $\{a_{\mathcal{K}}\}_{\mathcal{K} \subseteq [K]}$ values come from a valid distribution for $(A_1,\ldots, A_K)$, pointwise for every $(x^n,y_1^n,\ldots, y_K^n)$. To this end let us introduce the following lemma.
\begin{lemma} \label{lem:inclusion_exclusion}
	Let $\{a_{\mathcal{K}}\}_{\mathcal{K} \subseteq [K]}$ be a set of real values indexed by the subsets of $[K]$. Then there exists a realization of binary random variables $(A_1,A_2,\ldots, A_K) \sim \mP_{A_1\ldots A_K}\in \mathcal{P}(\{0,1\}^K)$ such that
	\begin{align}
		a_{\mathcal{K}} = \Pr( A_k=1, \forall k\in \mathcal{K}), ~~ \forall \mathcal{K} \subseteq [K],
	\end{align} 
	 if $a_{\emptyset} = 1$, and, for every $\mathcal{K} \subseteq [K]$, the resulting probability of the joint event $(A_k=1)_{k\in \mathcal{K}} \wedge (A_{k'}= 0)_{k'\in [K]\setminus \mathcal{K}}$ given by the inclusion-exclusion formula (\!\cite[Proposition 6]{Mobius})
	\begin{align}
		f_{\mathcal{K}} \triangleq \sum_{\mathcal{A} \subseteq [K]\setminus \mathcal{K}} (-1)^{|\mathcal{A}|} a_{\mathcal{A}\cup \mathcal{K}} 
	\end{align}
	is nonnegative.
\end{lemma}
\begin{proof}
	Taking the sum of $f_{\mathcal{K}}$ over all subsets $\mathcal{K} \subseteq [K]$, we have
	\begin{align}
		\sum_{\mathcal{K} \subseteq [K]} f_{\mathcal{K}} &= \sum_{\mathcal{K} \subseteq [K]}  \sum_{\mathcal{A} \subseteq [K]\setminus \mathcal{K}} (-1)^{|\mathcal{A}|} a_{\mathcal{A}\cup \mathcal{K}} \\
		&=\sum_{\mathcal{S} \subseteq [K]}\sum_{\mathcal{K} \subseteq \mathcal{S}} (-1)^{|\mathcal{S}|-|\mathcal{K}|}a_{\mathcal{S}} \label{eq:fixS}\\
		&= a_{\emptyset} \label{eq:cancel_except_emptyset}
	\end{align}
	Step \eqref{eq:fixS} follows because for each $\mathcal{S}\subseteq [K]$, the term $a_{\mathcal{S}}$ appears whenever $\mathcal{K}\subseteq \mathcal{S}$ and $\mathcal{A} = \mathcal{S} \setminus \mathcal{K}$. To see Step \eqref{eq:cancel_except_emptyset}, note that for any non-empty set $\mathcal{S}$, $\sum_{\mathcal{K} \subseteq \mathcal{S}} (-1)^{|\mathcal{S}|-|\mathcal{K}|} = (1-1)^{|\mathcal{S}|} = 0$.
	
	Therefore, if $a_{\emptyset} = 1$, the probability of those mutually exclusive events $(A_k=1)_{k\in \mathcal{K}} \wedge (A_{k'}= 0)_{k'\in [K]\setminus \mathcal{K}}$ for $\mathcal{K}\subseteq [K]$ add up to $1$, and since each of the probability $f_{\mathcal{K}}$ is non-negative, the existence of such $(A_1,\ldots, A_K)$ then follows. 
\end{proof}

We next group $(x^n,y_1,\ldots, y_K^n)$ into $2^K$ cases, each indexed by a unique subset $\mathcal{S} \subseteq [K]$, for which the predicate 
\begin{align*}
	\bigwedge_{s \in \mathcal{S}}\overline{I}_s(x^n,y_s^n) \bigwedge_{s'\in [K] \setminus \mathcal{S}} I_{s'}(x^n,y_{s'}^n)
\end{align*}
is true. 
In other words, $\mathcal{S}$ corresponds to the case for which authentication fails exactly at the users indexed by $k \in \mathcal{S}$.
For the case indexed by $\mathcal{S}$, we have
\begin{align} \label{eq:relation_ac}
	a_{\mathcal{K}} = c_{\mathcal{S}\cap \mathcal{K}}, ~~ \forall  \mathcal{K} \subseteq [K].
\end{align}
Applying Lemma \ref{lem:inclusion_exclusion} with this set of $\{a_{\mathcal{K}}\}$, we have that the condition for them to be realizable is,
\begin{align}
	\sum_{\mathcal{A} \subseteq [K]\setminus \mathcal{K}} (-1)^{|\mathcal{A}|} c_{\mathcal{S}\cap(\mathcal{A}\cup \mathcal{K})} \geq 0,~~\forall \mathcal{K} \subseteq [K].
\end{align}
Consider that each $\mathcal{A}\subseteq [K]\setminus \mathcal{K}$ in the summation is partitioned as $\mathcal{A} = \mathcal{I}\cup \mathcal{J}$ where $\mathcal{I} = \mathcal{A}\cap \mathcal{S}$ and $\mathcal{J} = \mathcal{A} \setminus \mathcal{I}= \mathcal{A}\cap (([K]\setminus \mathcal{K}) \setminus \mathcal{S})$. This is illustrated by the Venn diagram in Fig. \ref{fig:venn}.
\begin{figure}[htbp]
\centering
\begin{tikzpicture}[
	scale=0.9,
    transform shape,
    line join=round,
    line cap=round,
    >=stealth  
]
 
\begin{scope}
    \clip (0,0) ellipse[x radius=3, y radius=2];
    \fill[gray!25]
        (-0.35,3.0) --
        (5.0,3.0) --
        (5.0,-3.0) --
        (1.65,-3.0) --
        cycle;
\end{scope}

\begin{scope}
    \clip (0.5,0.75)
        ellipse[x radius=1.6, y radius=0.8];
    \fill[gray!10]
        (-0.35,3.0) --
        (1.65,-3.0) --
        (-4,-3.0) --
        (-4,3.0) --
        cycle;
\end{scope}

\begin{scope}
    \clip (0.5,0.75)
        ellipse[x radius=1.6, y radius=0.8];
    \fill[gray!65]
        (-0.35,3.0) --
        (4,3.0) --
        (4,-3.0) --
        (1.65,-3.0) --
        cycle;
\end{scope}

\draw[thick]
    (0,0) ellipse[x radius=3, y radius=2];

\begin{scope}
    \clip (0,0) ellipse[x radius=3, y radius=2];
    \draw[thick]
        (-3.6,-1.2)
        .. controls (-2,-0.1) and (2,-0.1) ..
        (3.6,-1.2);
\end{scope}
 
\draw[thick] (0.5,0.75)
    ellipse[x radius=1.6, y radius=0.8];

\node at (-2.1,-1) {$\mathcal{K}$};
\node at (-2.2,0) {$[K]\setminus\mathcal{K}$};

\node at (-0.3,0.7) {$\mathcal{J}$};
\node at (1.3,0.7) {$\mathcal{I}$};

\node at (-1.7,1.3) {$\mathcal{A}$};
\draw[->, thick] (-1.5,1.3) -- (-1,1.1);
\node at (4,1) {$\mathcal{S}$};
\draw[->,thick] (3.8,0.95) --  (2,0);
\end{tikzpicture}
\caption{A Venn diagram illustrating the sets $\mathcal{K}$, $[K]\setminus \mathcal{K}$, $\mathcal{S}$, $\mathcal{A}$, $\mathcal{I}$, and $\mathcal{J}$. The ambient set is $[K]$.}
\label{fig:venn}
\end{figure}
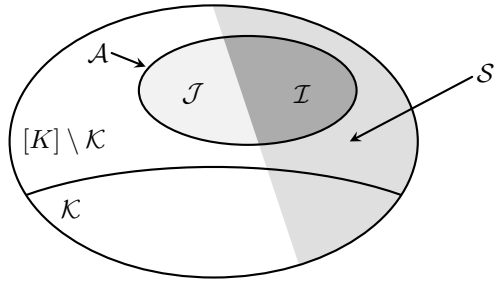
We have
\begin{align}
	&\sum_{\mathcal{A} \subseteq [K]\setminus \mathcal{K}} (-1)^{|\mathcal{A}|} c_{\mathcal{S}\cap(\mathcal{A}\cup \mathcal{K})}\notag \\
	&=\sum_{\mathcal{I}\subseteq ([K]\setminus \mathcal{K})\cap \mathcal{S}} \sum_{\mathcal{J} \subseteq ([K]\setminus \mathcal{K}) \setminus \mathcal{S}} (-1)^{|\mathcal{I}|+|\mathcal{J}|} c_{\mathcal{I}\cup (\mathcal{S}\cap \mathcal{K})}\\
	&=\sum_{\mathcal{J} \subseteq ([K]\setminus \mathcal{K}) \setminus \mathcal{S}} (-1)^{|\mathcal{J}|}  \sum_{\mathcal{I}\subseteq ([K]\setminus \mathcal{K}) \cap \mathcal{S}} (-1)^{|\mathcal{I}|} c_{\mathcal{I}\cup (\mathcal{S}\cap \mathcal{K})} \\
	&=\begin{cases}
		\sum_{\mathcal{I}\subseteq ([K]\setminus \mathcal{K}) \cap \mathcal{S}} (-1)^{|\mathcal{I}|} c_{\mathcal{I}\cup (\mathcal{S}\cap \mathcal{K})}, & ([K]\setminus \mathcal{K}) \subseteq \mathcal{S}  \\
		0, & ([K]\setminus \mathcal{K}) \not\subseteq \mathcal{S} \label{eq:cancel}
	\end{cases}
\end{align}
Therefore, we only need to focus on the cases for which $([K]\setminus \mathcal{K}) \subseteq \mathcal{S}$. Note that as $n\to \infty$,
\begin{align}
	c_{\mathcal{K}_2} = o(c_{\mathcal{K}_1})
\end{align} 
as long as $\mathcal{K}_1 \subset \mathcal{K}_2 \subseteq [K]$.  Therefore, as $n\to \infty$,
\begin{align}
	&\sum_{\mathcal{I}\subseteq ([K]\setminus \mathcal{K}) \cap \mathcal{S}} (-1)^{|\mathcal{I}|} c_{\mathcal{I}\cup (\mathcal{S}\cap \mathcal{K})} \simeq c_{\mathcal{S} \cap \mathcal{K}}.
\end{align}
Recall that $c_{\emptyset}=1$, and $c_{\mathcal{K}} \simeq \frac{1}{\prod_{k\in \mathcal{K}}M_k}$ for every non-empty $\mathcal{K} \subseteq [K]$.
It follows that 
\begin{align}
	\sum_{\mathcal{I}\subseteq ([K]\setminus \mathcal{K}) \cap \mathcal{S}} (-1)^{|\mathcal{I}|} c_{\mathcal{I}\cup (\mathcal{S}\cap \mathcal{K})}
\end{align}
is non-negative for all sufficiently large $n$.

Therefore, for every sufficiently small $\epsilon$, the above construction yields valid NS-assisted coding schemes $\mZ_{n}$ for all sufficiently large $n$.
 
\subsection{Reliability analysis}
Similar to the previous analysis for the $K=2$ user case, we now analyze the probability of error of the scheme  we described above for general $K$.
The transmitter provides the scheme with $(W_1,\ldots, W_K)$. The scheme generates $X^n$ (a random seed) according to $\mP_X^{\otimes n}$. The seed is used as the input to the broadcast channel, and for each $k\in [K]$, User $k$ receives $Y_k^n$ (a key), which comes from the stochastic mapping $\mN_{Y_k\mid X}^{\otimes n}$ of the seed. The key is then used as the input to the scheme at User $k$. The law of large numbers then implies
\begin{align}
	&\lim_{n\to \infty} \Pr(I_k(X^n,Y_k^n), \forall k\in [K])\\
	&=\lim_{n\to \infty}\Pr\big( (X^n,Y_k^n) \in \mathcal{T}^{(n)}_{\epsilon}(\mP_X\mN_{Y_k\mid X}), \forall k\in [K] \big) \\
	&= 1
\end{align}
It follows that with probability $\to 1$,
\begin{align}
	a_{[K]}(X^n, Y_1^n,\ldots, Y_K^n) = 1.
\end{align}
Therefore, with probability $\to 1$, the scheme provides the correct messages for all $K$ users. It follows that $\lim_{n\to \infty} P_{e,k}(\mZ_n)=0, \forall k\in [K]$.
The rate tuples $(R_1,\ldots, R_K)$ thus achieved are exactly described by Eq. \eqref{eq:achi_rates_K}. Taking $\epsilon\to 0$ concludes the proof. 
\hfil \qed

\section{Proof of Theorem \ref{thm:NSBC_Kuser}: Converse} \label{proof:converse}
Let us first prove the following lemma.
\begin{lemma}[NS data-processing] \label{lem:dp}
	Let $(W,X,Y, \widehat{W})\sim \mP_{WXY\widehat{W}} = \mS_W \mN_{Y\mid X} \mZ_{X\widehat{W}\mid WY}$, where $\mS_W$ is a pmf on $W$, $\mN_{Y\mid X}$ is a conditional pmf of $Y$ given $X$, and $\mZ_{X\widehat{W}\mid WY}$ is a conditional pmf satisfying the non-signaling conditions that $\mZ_{X\mid WY}(x\mid w,y) = \mZ_{X\mid W}(x\mid w)$ is independent of $y$ and $\mZ_{\widehat{W} \mid WY}(\widehat{w}\mid w,y) = \mZ_{\widehat{W}\mid Y}(\widehat{w}\mid y)$ is independent of $w$. Then 
	\begin{align}
		 I_{\mP}(X;Y) \geq I_{\mP}(W;\widehat{W}).
	\end{align}
\end{lemma}
\begin{proof}
	The proof applies the data-processing inequality used commonly in NS-assisted converses (cf. \cite{matthews2012linear, fawzi2024MAC}). Under $\mP_{WXY\widehat{W}}$, $(W,X,Y) \sim \mS_W \mZ_{X\mid W} \mN_{Y\mid X}$.  Let $\mP_Y$ be the marginal pmf thus obtained for $Y$, i.e., $\mP_Y(y) = \sum_{w,x}\mP_{W} \mZ_{X\mid W}(x\mid w) \mN_{Y\mid X}(y\mid x)$. Let $\mQ_{WXY\widehat{W}} = \mS_W \mP_Y \mZ_{X\widehat{W}\mid WY}$ be another distribution on $(W,X,Y,\widehat{W})$. Note that under $\mQ_{WXY\widehat{W}}$ one should think of $(W,Y)$ being independently generated according to $\mS_W \mP_Y$ and $(X,\widehat{W})$ are then generated according to $\mZ_{X\widehat{W}\mid WY}$. Then $\mQ_{W \widehat{W}} = \mS_W \mQ_{\widehat{W}}=\mP_W \mQ_{\widehat{W}}$, as $\sum_{x,y}\mQ(w,x,y,\widehat{w}) = \sum_{x,y}\mS(w) \mP_Y(y) \mZ(x,\widehat{w}\mid x,y)= \sum_{y}\mS(w) \mP_Y(y) \mZ(\widehat{w}\mid y) = \mS(w) \sum_y  \mQ(y,\widehat{w})  = \mS(w) \mQ(\widehat{w})$. According to the data-processing inequality for relative entropy, $D(\mP_{WXY\widehat{W}} \Vert \mQ_{WXY\widehat{W}}) \geq D(\mP_{W\widehat{W}} \Vert \mQ_{W\widehat{W}}) = D(\mP_{W\widehat{W}} \Vert \mP_{W}\mQ_{\widehat{W}})=I_{\mP}(W;\widehat{W})+D(\mP_{\widehat{W}}\Vert Q_{\widehat{W}}) \geq I_{\mP}(W;\widehat{W})$. On the other hand, $D(\mP_{WXY\widehat{W}} \Vert \mQ_{WXY\widehat{W}}) = \sum_{w,x,y,\widehat{w}}\mP(w,x,y,\widehat{w}) \log_2(\frac{\mN_{Y\mid X}(y\mid x)}{\mP_Y(y)}) = I_{\mP}(X;Y)$.  
\end{proof}

We shall prove that, if $(R_1,R_2,\ldots, R_K)\in \mathbb{R}_{\geq 0}^K$ is achievable by NS-assisted coding schemes, then $(R_1,R_2,\ldots, R_K) \in \mathcal{R}_{\Sato}(\mN_{Y_1Y_2\cdots Y_K})$. In the following, for $\mathcal{K} \subseteq [K]$, we write $A_{\mathcal{K}} \triangleq (A_k)_{k\in \mathcal{K}}$. 
Let $(R_1,R_2,\ldots, R_K)$ be any rate tuple achievable by NS-assisted coding schemes. By definition, there exists a sequence of $(M_1^{(n)},M_2^{(n)},\ldots, M_K^{(n)},n)$ NS-assisted coding schemes $(\mZ_n)_{n\in \mathbb{N}}$, such that $\lim_{n\to \infty} P_{e,k}(\mZ_n) = 0$ and $\liminf_{n\to \infty}\frac{1}{n} \log_2 M_k^{(n)} \geq R_k, \forall k\in [K]$.

For each non-empty subset $\mathcal{K} \subseteq  [K]$, consider the users indexed by $k\in \mathcal{K}$. Let $p_{e,\mathcal{K}}^{(n)} \triangleq \Pr(\widehat{W}_{\mathcal{K}} \neq W_{\mathcal{K}})$. Let $H_b(\cdot)$ denote the binary entropy function. We have
\begin{align}
	&(1-p_{e,\mathcal{K}}^{(n)}) \sum_{k\in \mathcal{K}}\log_2 M_k^{(n)} - H_b(p_{e,\mathcal{K}}^{(n)}) \notag \\
	&\leq I(W_{\mathcal{K}}; \widehat{W}_{\mathcal{K}}) \label{eq:use_fano} \\
	&\leq I(W_{[K]}; \widehat{W}_{\mathcal{K}})\\
	&\leq I(X^n; Y_{\mathcal{K}}^n) \label{eq:use_dp}
\end{align}
where Step \eqref{eq:use_fano} follows from Fano's inequality, and Step \eqref{eq:use_dp} follows by applying Lemma \ref{lem:dp} with $W = W_{[K]}$, $X=X^n$, $Y = Y_{\mathcal{K}}^n$, $\widehat{W} = \widehat{W}_{\mathcal{K}}$. 
Since $\lim_{n\to \infty} P_{e,k}(\mZ_n) = 0$ for all $k\in \mathcal{K}$, we have  $\lim_{n\to \infty}p_{e,\mathcal{K}}^{(n)} = 0$.   Therefore, 
\begin{align}
	n\sum_{k\in \mathcal{K}} R_k-o(n) \leq  I(X^n;Y_{\mathcal{K}}^n).
\end{align}

Due to the same-marginals property (Remark \ref{rem:same_marginals}), the probability of decoding error for each user with index in $\mathcal{K}$ will not change if we replaced their channel $\mN_{Y_{\mathcal{K}}\mid X}$ by any channel $\widetilde{\mN}_{\widetilde{Y}_{\mathcal{K}}\mid X}$ such that $\widetilde{\mN}_{\widetilde{Y}_{k}\mid X} = \mN_{Y_k\mid X}, \forall k\in \mathcal{K}$. 
Therefore, for any such channel $\widetilde{\mN}_{\widetilde{Y}_{\mathcal{K}}\mid X}$ and $(\widetilde{Y}_{\mathcal{K}}^n, \widetilde{W}_{\mathcal{K}})$ such defined, the same $(\mZ_n)_{n\in \mathbb{N}}$ also yields that $\lim_{n\to \infty}P_{e,k}(\mZ_n) =0, \forall k\in \mathcal{K}$, and therefore $\lim_{n\to \infty}p_{e,\mathcal{K}}^{(n)} =0, \forall k\in \mathcal{K}$. 
Applying the same argument for each of such $\widetilde{\mN}_{\widetilde{Y}_{\mathcal{K}}\mid X}$, we have that
\begin{align}
	n\sum_{k\in \mathcal{K}} R_k - o(n) \leq \min_{\widetilde{\mN}_{\widetilde{Y} _{\mathcal{K}}\mid X}}  I(X^n;\widetilde{Y}_{\mathcal{K}}^n),
\end{align}
where the minimization is over all channels $\widetilde{\mN}_{\widetilde{Y}_{\mathcal{K}}\mid X}$ with the same $|\mathcal{K}|$ conditional marginal pmfs as $\mN_{Y_{\mathcal{K}}\mid X}$. 

Therefore, for each non-empty subset $\mathcal{K} \subseteq [K]$ and $n\in \mathbb{N}$,
\begin{align}
	&n\sum_{k\in \mathcal{K}} R_k -o(n) \notag\\
	 &\leq \min_{\widetilde{\mN}_{\widetilde{Y}_{\mathcal{K}}\mid X}}   I(X^n;\widetilde{Y}_{\mathcal{K}}^n) \\
	&= \min_{\widetilde{\mN}_{\widetilde{Y}_{\mathcal{K}}\mid X}}  \sum_{i=1}^nI(X^n;  \widetilde{Y}_{\mathcal{K},i}  \mid \widetilde{Y}_{\mathcal{K}}^{i-1})\\
  &\leq \min_{\widetilde{\mN}_{\widetilde{Y}_{\mathcal{K}}\mid X}}  \sum_{i=1}^n I(X_i; \widetilde{Y}_{\mathcal{K},i}) \label{eq:use_Markov1}\\
  &= n \min_{\widetilde{\mN}_{\widetilde{Y}_{\mathcal{K}}\mid X}}   I(X_{T}; \widetilde{Y}_{\mathcal{K},T} \mid T) \\
  &\leq n \min_{\widetilde{\mN}_{\widetilde{Y}_{\mathcal{K}}\mid X}}  I(X_{T}; \widetilde{Y}_{\mathcal{K},T}) \label{eq:use_Markov2}\\
  &= n J(\mP_{X_T},(\mN_{Y_k\mid X})_{k\in \mathcal{K}}) \label{eq:introduce_PXT}
\end{align}
where in Step \eqref{eq:use_Markov1} we introduce $T$ as a time-sharing variable uniformly distributed over $[n]$ and is independent of $(W_{[K]},X^n,\widetilde{Y}_{\mathcal{K}})$. 
Step \eqref{eq:use_Markov1} follows from the Markov chain $(X^n,\widetilde{Y}_{\mathcal{K}}^{i-1}) \leftrightarrow X_i \leftrightarrow \widetilde{Y}_{\mathcal{K},i})$. Step \eqref{eq:use_Markov2} follows from the Markov chain $T\leftrightarrow X_{T} \leftrightarrow \widetilde{Y}_{\mathcal{K},T}$. In Step \eqref{eq:introduce_PXT}, $\mP_{X_T}\in \mathcal{P}(\mathcal{X})$ is the distribution for $X_T$. Note that, for each $n$, the distribution $\mP_{X_T}$ is the same for different $\mathcal{K}$ because it is completely determined by $\mZ_n$, independent of the channel. Taking $n\to \infty$ concludes the proof.  \hfill \qed

\section{Proof of Corollary \ref{cor:BSSC}} \label{proof:BSSC}
Let $\mN_{Y_1Y_2\mid X}$ be the BSSC. Let us first show that $\max_{(R_1,R_2) \in \mathcal{C}^{\NS}}(R_1+R_2) \geq 0.5$. For this let us focus on the region $\mathcal{R}_{\Sato}(\mN_{Y_1Y_2\mid X})$.
Note that the same-marginals class of BSSC contains only one channel, i.e., $|\Gamma(\mN_{Y_1\mid X}, \mN_{Y_2\mid X})| = 1$. This facilitates the evaluation of $J(\mP_X,(\mN_{Y_1\mid X}, \mN_{Y_2\mid X}))$ since there is no minimization to be done. Let $H_b(x)=-x\log_2(x)-(1-x)\log_2(1-x)$ denote the binary entropy function, with the convention that $0\log_2(0) = 0$.
With $\mP_X$ chosen to be the uniform distribution over $\{0,1\}$, we obtain $I(X;Y_k) = H_b(\frac{1}{4})-\frac{1}{2} \approx 0.31, \forall k\in \{1,2\}$ and $I(X;Y_1,Y_2) = 0.5$. It follows that $(0.25,0.25)\in \mathcal{C}^{\NS}(\mN_{Y_1Y_2\mid X})$ and $\max_{(R_1,R_2)\in \mathcal{C}^{\NS}} (R_1+R_2) \geq 0.5$.

Let us then show that $\max_{(R_1,R_2)\in \mathcal{R}_{\KSp}} (R_1+R_2) < 0.45$. Firstly, due to the symmetry in the BSSC, we must have that $\max_{(R_1,R_2) \in \mathcal{R}_{\KS,1}}(R_1+R_2) = \max_{(R_1,R_2) \in \mathcal{R}_{\KS,2}}(R_1+R_2)$. Then,  by the definition of convex hull, any point $(R_1,R_2)\in \mathcal{R}_{\KSp}$ can be written as $(R_1,R_2) = \lambda(R_1',R_2') + (1-\lambda)(R_1'',R_2'')$ where $(R_1',R_2') \in \mathcal{R}_{\KS,1}$, $(R_1'',R_2'') \in \mathcal{R}_{\KS,2}$, and $\lambda \in [0,1]$. Therefore $R_1+R_2 = \lambda(R_1'+R_2') + (1-\lambda) (R_1''+R_2'') \leq \max \{R_1'+R_2', R_1''+R_2''\}$. It follows that $\max_{(R_1,R_2) \in \mathcal{R}_{\KSp}}(R_1+R_2) = \max_{(R_1,R_2) \in \mathcal{R}_{\KS,1}} (R_1+R_2)$. Now let us focus on $\mathcal{R}_{\KS,1}(\mN_{Y_1Y_2\mid X})$. Note that for any $\mP_{XU}\in \mathcal{P}(\mathcal{X} \times \mathcal{U})$, $I(X;Y_1) +I(U;Y_2) \geq I(X;Y_1\mid U)+I(U;Y_2)$ because $U\leftrightarrow X \leftrightarrow Y_1$ form a Markov chain. In other words, the sum-rate bound is always tight for $R_1+R_2$. Therefore, $\max_{(R_1,R_2)\in \mathcal{R}_{\KS, 1}} (R_1+R_2) = \max_{\mP_{XU}} \big(I(X;Y_1\mid U)+I(U;Y_2)\big)$. 

Parameterize $\mP_{XU}$ such that for each $u\in \mathcal{U}$, $\mP_U(u) = p_u, \mP_{X\mid U}(1\mid u) \triangleq q_u$, and $\mP_{X}(1) \triangleq q$. Note that for each $u\in \mathcal{U}$, $I(X;Y_1\mid U=u) = H_b(\frac{q_u}{2})-q_u$, and  $I(X;Y_2\mid U=u) = H_b(\frac{1-q_u}{2})-(1-q_u)$. Also, $I(X;Y_2) = H_b(\frac{1-q}{2})-(1-q)$. Then the objective
\begin{align}
	&I(X;Y_1\mid U) + I(U;Y_2) \notag \\
	&= I(X;Y_1\mid U) - I(X;Y_2\mid U) + I(X;Y_2) \label{eq:use_markov_BSSC} \\
	&= \sum_{u\in \mathcal{U}} p_u \Big( H_b(\frac{q_u}{2})-q_u \Big) - \sum_{u\in \mathcal{U}} p_u \Big( H_b(\frac{1-q_u}{2})-(1-q_u) \Big) + \Big(H_b(\frac{1-q}{2})-(1-q)\Big) \label{eq:to_bound}
\end{align}
where Step \eqref{eq:use_markov_BSSC} is because $U\leftrightarrow X \leftrightarrow Y_2$ forms a Markov chain. 

For $x\in [0,1]$, define 
\begin{align}
	f(x) &\triangleq H_b(\frac{x}{2})-x \\
	&= -\frac{x}{2}\log_2(\frac{x}{2})-(1-\frac{x}{2})\log_2(1-\frac{x}{2})
\end{align}
and 
\begin{align}
	g(x) &\triangleq f(x)-f(1-x) \\
	&= -\frac{x}{2}\log_2(\frac{x}{2}) - (1-\frac{x}{2})\log_2(1-\frac{x}{2}) \notag \\
	&~~~~+ \frac{1-x}{2} \log_2(\frac{1-x}{2}) + \frac{1+x}{2}\log_2(\frac{1+x}{2}) \\
	&~~~~- 2x +1
\end{align}
It can be shown that $g(x) < \frac{1}{9}, \forall x\in [0,1]$, and that $f(x)<1/3, \forall x\in [0,1]$. It follows that
\begin{align}
	\eqref{eq:to_bound} &= \sum_{u\in \mathcal{U}} p_u \big( f(q_u)-f(1-q_u) \big) + f(1-q) \\
	&= \sum_{u\in \mathcal{U}} p_u g(q_u) + f(1-q) \\
	&<\frac{1}{9}+\frac{1}{3} \\
	&<0.45
\end{align}

\section*{Acknowledgment}
The authors acknowledge the assistance of ChatGPT 5.6 in this work. In particular, after the authors pointed ChatGPT 5.6 to the authentication solution in \cite{Yao_Jafar_NSCWS}, a series of discussions proceeding through toy examples led to an initial achievability proof for the $K=2$ case. Following a sequence of simplifications and elimination of unnecessary steps, the extension to arbitrary $K$ was obtained by recognizing the underlying pattern revealed by the refined $K=2$ construction. ChatGPT 5.6 also assisted in deriving the proof of Lemma \ref{lem:typical_off_K} without the additional $o(n)$ term in the exponent, identifying the BSSC example for which the inclusion is strict, drawing the figures that appear in the paper, and literature search. The authors developed the final proofs and accept full responsibility for the accuracy and originality of the results.

\bibliographystyle{IEEEtran}
\bibliography{../../bib_file/yy.bib}
\end{document}